\documentclass[10pt,reqno]{amsart}

\usepackage[applemac]{inputenc}
\usepackage[english]{babel}
\usepackage[babel=true]{csquotes}
\usepackage{amsfonts,amsmath,amssymb}
\usepackage{indentfirst}
\usepackage{bm}
\usepackage{bbm}
\usepackage{xcolor}
\usepackage{amsthm}
\usepackage{colortbl}
\usepackage{enumerate}
\usepackage[cm]{aeguill}
\usepackage{graphicx}
\usepackage{url}
\usepackage{xcolor}
\usepackage{dsfont}
\usepackage{braket}
\usepackage{appendix}
\usepackage{upgreek}

\newtheorem{thm}{Theorem}[section]
\newtheorem{cor}[thm]{Corollary}
\newtheorem{lem}[thm]{Lemma}

\theoremstyle{remark}
\newtheorem{rem}[thm]{Remark}
\numberwithin{equation}{section}
\renewcommand{\theenumi}{\roman{enumi}.}

\makeatletter
\newenvironment{enumeratenew}[1][\roman]{%
\renewcommand{\theenumi}{#1{enumi}.}\begin{enumerate}}
{\end{enumerate}}
\makeatother
\newcommand\Tline{\rule{0pt}{3.1ex}}
\newcommand\Bline{\rule[-1.7ex]{0pt}{0pt}}
\newcommand{\tr}{\operatorname{tr}}

\begin{document}

\title[The relativistic mean-field equations of the atomic nucleus]
 {The relativistic mean-field equations of the atomic nucleus}

\author{Simona Rota Nodari}

\address{CNRS, Laboratoire Jacques-Louis Lions (UMR 7598)\newline
Université Pierre et Marie Curie (Paris VI)\newline
Boîte courrier 187\newline
75252 Paris Cedex 05
France}

\email{rotanodari@ann.jussieu.fr}


\date{\today}

\begin{abstract} 
In nuclear physics, the relativistic mean-field theory describes the nucleus as a system of Dirac nucleons which interact via meson fields. In a static case and without nonlinear self-coupling of the $\sigma$ meson, the relativistic mean-field equations become a system of Dirac equations where the potential is given by the meson and photon fields. The aim of this work is to prove the existence of solutions of these equations. We consider a minimization problem with constraints that involve negative spectral projectors and we apply the concentration-compactness lemma to find a minimizer of this problem. We show that this minimizer is a solution of the relativistic mean-field equations considered.
\end{abstract}

\maketitle

\section{Introduction}

In this paper, we present the first mathematically rigorous result concerning the existence of solutions of the relativistic mean-field equations of the atomic nucleus in a static case and without nonlinear self-coupling of the $\sigma$ meson.

Though often used in practice, the models of nuclear physics have rarely been considered from a mathematical point of view: some nonrelativistic models (of Hartree-Fock type) were studied by D. Gogny and P.L. Lions in 1986 (\cite{gognylions}), but, to our knowledge, there are no rigorous mathematical studies of relativistic models, which are however extremely important in nuclear physics.

In nuclear physics, the relativistic mean-field RMF theory describes the nucleus as a system of Dirac nucleons which interact in a relativistic covariant manner via meson fields. During the last years, the relativistic mean-field theory has received wide attention due to its successful description of lots of nuclear phenomena. The
relativistic mean-field model is considered to be the relativistic generalization of the nonrelativistic models such as the Skyrme force or the Gogny force Hartree-Fock theory, using effective mesonic degrees of freedom rather than instantaneous forces.
The relativistic model describes successfully the single-particle structure of nuclei as the nonrelativistic ones and provides a natural explanation of some relativistic effects as the spin-orbit force (see  \cite{reinhard},\cite{greinermaruhn},\cite{ring},\cite{meng}).

The model is formulated on the basis of two approximations, the mean-field and the no-sea approximation. Thanks to the mean-field approximation, the fields for the mesons and the photons are treated as classical fields and the nucleons behave as noninteracting particles moving in these mean fields. This implies that the  nucleon field operator can be expanded in single-particle states $\uppsi_\alpha\left(x^\mu\right)$,
\begin{equation}\label{eqnucleonstate}
\uppsi=\sum_\alpha\uppsi_\alpha\left(x^\mu\right)\hat{a}_\alpha
\end{equation}
where $\hat{a}_\alpha$ is the annihilation operator for a nucleon in the state $\alpha$, while the densities become simple bilinear sums over the $\uppsi_\alpha$. The no-sea approximation corresponds to neglecting the vacuum polarization, that means that we have a number of occupied single-particle orbitals $\uppsi_\alpha$, $\alpha=1,\ldots,\Omega$, which determines the densities. We remind that when the isospin\footnote{Isospin (contraction of  isotopic spin) 
is a quantum number related to the strong interaction. Isospin was introduced by Heisenberg in 1932; he observed that the neutron is almost identical to the proton, apart from the fact that it carries no charge. In particular, their masses are close and they are indistinguishable under the strong interactions. So, the proton and the neutron appear to be two states of the same particle, the nucleon, associated with different isospin projections (\cite{griffiths},\cite{greinermaruhn}).} of the particles is not fixed, $\uppsi_\alpha\left(x^\mu\right)\in \mathbb{C}^2\otimes\mathbb{C}^4$. Moreover, the single-particle wave functions have to satisfy the constraint $\int_{\mathbb{R}^3}\uppsi_\alpha^*(t,x)\uppsi_\beta(t,x)\,d^3x=\delta_{\alpha\beta}$.

The Lagrangian density of the RMF theory can be written as
\begin{equation}\label{eqlagrangian}
\mathcal{L}=\mathcal{L}_{nucleons}+\mathcal{L}_{mesons}+\mathcal{L}_{coupling}.
\end{equation}
The free Lagrangian for the nucleons is 
\begin{equation}\label{eqlagrangiannu}
\mathcal{L}_{nucleons}=\sum_{\alpha=1}^{\Omega}w_\alpha{\bar{\uppsi}_\alpha}(i\left(\mathbbm{1}_2\otimes\gamma^{\mu}\right)\partial_{\mu}-m_b){\uppsi_\alpha}
\end{equation}
where $m_b$ denotes the nucleon mass, $\gamma^{\mu}$ are the Dirac matrices, $w_\alpha$ are occupation weights, $0\leq w_\alpha\leq1$, and $\bar{\uppsi}_\alpha={\uppsi}^*_\alpha\left(\mathbbm{1}_2\otimes\gamma^{0}\right)$ with $\mathbbm{1}_2=\left(\begin{array}{cc}1&0\\0&1\end{array}\right)$. 

The Lagrangian for the free meson fields is
\begin{eqnarray}\label{eqlagrangianme}
\mathcal{L}_{mesons}&=&\frac{1}{2}(\partial^{\mu}{\sigma}\partial_{\mu}{\sigma}-m^2_{\sigma}{\sigma}^2)\nonumber\\
&&-\frac{1}{2}(\overline{\partial^{\mu}{\omega}^{\nu}}\partial_{\mu}{\omega}_{\nu}-m^2_{\omega}{\omega}^{\mu}{\omega}_{\mu})\nonumber\\
&&-\frac{1}{2}(\overline{\partial^{\mu}{\bm{R}}^{\nu}}\cdot\partial_{\mu}{\bm{R}}_{\nu}-m^2_{\rho}{\bm{R}}^{\mu}\cdot{\bm{R}}_{\mu})\nonumber\\
&&-\frac{1}{2}\overline{\partial^{\mu}{A}^{\nu}}\partial_{\mu}{A}_{\nu}
\end{eqnarray}
where ${\sigma}$, ${\omega}^{\mu}$ and ${\bm{R}}^{\mu}$ describe respectively the $\sigma$, $\omega$ and $\rho$ meson field, and ${A}^{\mu}$ stands for the photon field. Moreover, an antisymmetrized derivative is defined via
$$
\overline{\partial^{\mu}{A}^{\nu}}=\partial^{\mu}{A}^{\nu}-\partial^{\nu}{A}^{\mu}.
$$
We remind that the $\sigma$ meson is an isoscalar scalar meson which provides a medium range attractive interaction, the $\omega$ meson is an isoscalar vector meson leading to a short range repulsive interaction, the $\rho$ meson is an isovector vector  meson needed for a better description of isospin-dependent effects in the nuclei, and the photon describes the electromagnetic interaction. \\
Finally, the Lagrangian for the coupling is
\begin{equation}\label{eqlagrangiancou}
\mathcal{L}_{coupling}=-g_{\sigma}{\sigma}{\rho}_s-g_{\omega}{\omega}^{\mu}{\rho}_{\mu}-g_{\rho}{\bm{R}}^{\mu}\cdot{\bm{\rho}}_{\mu}-e{A}^{\mu}{\rho}_{\mu}^c-U\left({\sigma}\right)
\end{equation}
where $U\left({\sigma}\right)=\frac{1}{3}b_2{\sigma}^3+\frac{1}{4}b_3{\sigma}^4$ represents a nonlinear self-coupling of the $\sigma$ meson. Note that in Reinhard's paper \cite{reinhard} the coupling constant for the $\rho$ meson is $2g_\rho$.\\
The densities are
\begin{eqnarray}\label{eqdensities}
\rho_s&=&\sum_{\alpha=1}^{\Omega}w_\alpha\bar{\uppsi}_\alpha\uppsi_\alpha,\\
\rho_\mu&=&\sum_{\alpha=1}^{\Omega}w_\alpha\bar{\uppsi}_\alpha\left(\mathbbm{1}_2\otimes\gamma_\mu\right)\uppsi_\alpha,\\
\bm{\rho}_\mu&=&\sum_{\alpha=1}^{\Omega}w_\alpha\bar{\uppsi}_\alpha\left(\hat{\bm{\tau}}\otimes\gamma_\mu\right)\uppsi_\alpha,\\
\rho_\mu^c&=&\sum_{\alpha=1}^{\Omega}w_\alpha\bar{\uppsi}_\alpha\left(\frac{1}{2}(\mathbbm{1}_2+\hat\tau_0)\otimes\gamma_\mu\right)\uppsi_\alpha.
\end{eqnarray}
We remind that ${\bm{R}}$ and $\bm{\rho}$ are vectors in isospin space and $\cdot$ denotes the vector product therein, and $\hat{\bm \tau}$ is the vector of the Pauli matrices which occurs in the definition of the isospin operator. More precisely, the three components of the isospin operator are defined by $\hat{\bm t}=\frac{1}{2}\hat{\bm \tau}$ and, in particular, the third component is given by
 $$\hat t_0=\frac{1}{2}\hat \tau_0=\frac{1}{2}\left(\begin{array}{cc}{1} & 0\\ 0 & -{1}\end{array}\right);$$
 the proton state, represented by the vector $\left(\begin{array}{c}1\\0\end{array}\right)$, is the eigenstate of $\hat\tau_0$ associated with the eigenvalue $\tau_0=1$ and the neutron state, represented by the vector $\left(\begin{array}{c}0\\1\end{array}\right)$, is the eigenstate of $\hat\tau_0$ associated with the eigenvalue  $\tau_0=-1$.

The model contains as free parameters the meson masses $m_\sigma$, $m_\omega$ and $m_\rho$, as well as the coupling constants $g_\sigma$, $g_\omega$, $g_\rho$, $b_2$ and $b_3$. For the nucleon mass $m_b$ the free value is usually employed.

Most applications of the relativistic mean-field model  are concerned with stationary states; then, like in \cite{reinhard}, we want to derive the field equations for the static case. Moreover,
we remark that it is generally true that proton and neutron states do not mix, that means that the single-particle states are eingenstates of the operator $\hat\tau_0$. As a consequence, only the components with isospin projection $0$ appear, i.e. $R_{0\mu}$ and $\rho_{0\mu}$.

Stationarity implies that all time derivatives and also the spatial components of densities and fields vanish; only the fields $\sigma$, $\omega_0$, $R_{00}$ and $A_0$ remain and they are independent of time. 

Furthermore, the single-particle wave functions can be written as $\uppsi_\alpha=\left(\begin{array}{c}1\\0\end{array}\right)\otimes\psi_\alpha$ for protons, and $\uppsi_\alpha=\left(\begin{array}{c}0\\1\end{array}\right)\otimes\psi_\alpha$ for neutrons. Each function $\psi_\alpha$ may be separated as
\begin{equation}\label{eqwave}
\psi_\alpha(t,x)=e^{-i\varepsilon_{\alpha}t}\psi_\alpha({x})
\end{equation}
where the $\varepsilon_\alpha$ are the single-particle energies and $\varepsilon_\alpha>0$.

Varying the action integral $S=\int \mathcal{L}\,d^4x$ with respect to the wave functions and to the fields with all the above simplifications inserted yields
\begin{eqnarray}\label{eqnucleardirac}
\varepsilon_\alpha\gamma_0\psi_\alpha&=&\left[-i\bm{\gamma}\cdot\nabla+m_b+g_\sigma\sigma+g_\omega\omega_0\gamma_0\right.\nonumber\\
&&\left.+g_\rho R_{00}\gamma_0\tau_0+\frac{1}{2}eA_0\gamma_0(1+\tau_0)\right]\psi_\alpha,\\
\label{eqnuclearsigma}
(-\Delta+m^2_\sigma)\sigma +U'(\sigma)&=&-g_\sigma\rho_s,\\
\label{eqnuclearomega}
(-\Delta+m^2_\omega)\omega_0 &=&g_\omega\rho_0,\\
\label{eqnuclearR}
(-\Delta+m^2_\rho)R_{00} &=&g_\rho\rho_{00},\\
\label{eqnuclearA}
-\Delta A_0 &=&e\rho_0^c,
\end{eqnarray} 
with $\bm{\gamma}=(\gamma^1,\gamma^2,\gamma^3)$. This set of equations, together with the definition of the densities, constitutes a self-consistent field problem that can be solved numerically using an iterative scheme (see \cite{reinhard}, \cite{bottcher}). We observe that there is no proof of convergence of this algorithm.

In this paper, we consider the case without nonlinear self-coupling of the $\sigma$ meson, i.e. $b_2=b_3=0$, and we choose a fixed occupation, that means that the occupation weights $w_\alpha$ are defined as 
\begin{equation}\label{eqoccupation}
w_\alpha=\left \{\begin{array}{ll}1 & \alpha=1,\ldots,A\\[5pt]
0& \mbox{otherwise}
\end{array}
\right.
\end{equation}
where $A$ is the nucleon number. In this case, the equations  (\ref{eqnuclearsigma}-\ref{eqnuclearA}) can be solved explicitly and we obtain
\begin{eqnarray}
\label{eqsigma}
\sigma&=&-\frac{g_\sigma}{4\pi}\left(\frac{e^{-m_\sigma|\cdot|}}{|\cdot|}\star\rho_s\right),\\
\label{eqomega}
\omega_0 &=&\frac{g_\omega}{4\pi}\left(\frac{e^{-m_\omega|\cdot|}}{|\cdot|}\star\rho_0\right),\\
\label{eqR}
R_{00} &=&\frac{g_\rho}{4\pi}\left(\frac{e^{-m_\rho|\cdot|}}{|\cdot|}\star\rho_{00}\right),\\
\label{eqA}
A_0 &=&\frac{e}{4\pi}\left(\frac{1}{|\cdot|}\star\rho_0^c\right).
\end{eqnarray} 
Hence, the equation (\ref{eqnucleardirac}) becomes
\begin{eqnarray}\label{eqdiracyukawa}
\varepsilon_\alpha\psi_\alpha&=&\left[H_0-\beta \frac{g_\sigma^2}{4\pi}\left(\frac{e^{-m_\sigma|\cdot|}}{|\cdot|}\star\rho_s\right)+\frac{g_\omega^2}{4\pi}\left(\frac{e^{-m_\omega|\cdot|}}{|\cdot|}\star\rho_0\right)\right.\nonumber\\
&&\left.+\tau_0\frac{g_\rho^2}{4\pi}\left(\frac{e^{-m_\rho|\cdot|}}{|\cdot|}\star\rho_{00}\right)+\frac{1}{2}(1+\tau_0)\frac{e^2}{4\pi}\left(\frac{1}{|\cdot|}\star\rho_0^c\right)\right]\psi_\alpha
\end{eqnarray}
where $H_0=-i\bm{\alpha}\cdot\nabla+\beta m_b$ is the free Dirac operator, $$\begin{array}{cc}\beta=\left(\begin{array}{cc}\mathbbm{1} & 0\\ 0 & -\mathbbm{1}\end{array}\right),& \alpha_k=\left(\begin{array}{cc}0 & \sigma_k\\ \sigma_k & 0\end{array}\right)\end{array}$$ for $k=1,2,3$, with
$$\begin{array}{ccc}\sigma_1=\left(\begin{array}{cc}0 & 1\\ 1 & 0\end{array}\right),& \sigma_2=\left(\begin{array}{cc}0 & -i\\ i & 0\end{array}\right),&\sigma_3=\left(\begin{array}{cc}{1} & 0\\ 0 & -{1}\end{array}\right).\end{array}$$ 
The operator $H_0$ acts on $4$-spinors, i.e. functions $\psi\in \mathcal H :=L^2(\mathbb R^3,\mathbb C^4)$. It is self-adjoint on $\mathcal H$, with domain $H^1(\mathbb R^3,\mathbb C^4)$ and form-domain $E:= H^{1/2}(\mathbb R^3,\mathbb C^4)$. Moreover, it is defined to ensure
$$
H^2_0=-\Delta +m_b^2.
$$
The spectrum of $H_0$ is $(-\infty,-m_b]\cup[m_b,+\infty)$, and the projector associated with the negative (resp. positive) part of the spectrum of $H_0$ will be denoted by $\Lambda^-$ (resp. $\Lambda^+$). Finally, we endow the space $E$ with the norm $\|\psi\|^2_{E}:=(\psi,\left|H_0\right|\psi)_{L^2}$.

Using the convention $\tau_0=1$ for the protons and $\tau_0=-1$ for the neutrons, the densities can be written as
\begin{eqnarray}\label{eqdensitiesfinal}
\rho_s&=&\sum_{k=1}^{A}\bar{\psi}_k\psi_k,\\
\rho_0&=&\sum_{k=1}^{A}{\psi}_k^\ast\psi_k,\\
\rho_{00}&=&\sum_{k=1}^{Z}{\psi}_k^\ast\psi_k-\sum_{k=Z+1}^{A}{\psi}_k^\ast\psi_k,\\
\rho_0^c&=&\sum_{k=1}^{Z}\psi_k^\ast\psi_k
\end{eqnarray}
with $Z$ the number of protons, $N=A-Z$ the number of neutrons and $\bar{\psi}_i={\psi}_i^\ast\beta$; furthermore, the nonlinear Dirac equations are given by
\begin{eqnarray}\label{eqdiracyukawaproton}
H_{ p,\Psi}\psi_{i}&:=&\left[H_0-\beta \frac{g_\sigma^2}{4\pi}\left(\frac{e^{-m_\sigma|\cdot|}}{|\cdot|}\star\rho_s\right)+\frac{g_\omega^2}{4\pi}\left(\frac{e^{-m_\omega|\cdot|}}{|\cdot|}\star\rho_0\right)\right.\nonumber\\
&&\left.+\frac{g_\rho^2}{4\pi}\left(\frac{e^{-m_\rho|\cdot|}}{|\cdot|}\star\rho_{00}\right)+\frac{e^2}{4\pi}\left(\frac{1}{|\cdot|}\star\rho_0^c\right)\right]\psi_{i}=\varepsilon_{i}\psi_{i}
\end{eqnarray}
if $1\le i\le Z$, and 
\begin{eqnarray}\label{eqdiracyukawaneutron}
H_{ n,\Psi}\psi_{i}&:=&\left[H_0-\beta \frac{g_\sigma^2}{4\pi}\left(\frac{e^{-m_\sigma|\cdot|}}{|\cdot|}\star\rho_s\right)+\frac{g_\omega^2}{4\pi}\left(\frac{e^{-m_\omega|\cdot|}}{|\cdot|}\star\rho_0\right)\right.\nonumber\\
&&\left.-\frac{g_\rho^2}{4\pi}\left(\frac{e^{-m_\rho|\cdot|}}{|\cdot|}\star\rho_{00}\right)\right]\psi_i=\varepsilon_i\psi_i
\end{eqnarray} 
if $Z+1\le i\le A$, with $\Psi=\left(\psi_1,\ldots,\psi_Z,\psi_{Z+1},\ldots,\psi_{A}\right)$ and under the constraints $\int_{\mathbb{R}^3}\psi^\ast_i\psi_j=\delta_{ij}$ for $1\leq i,j\leq Z$ and for $Z+1\leq i,j \leq A$. \\ 
In what follows, $V_{p,\Psi}$ and $V_{n,\Psi}$ denote the potentials  of the nonlinear Dirac equations, namely $V_{\mu,\Psi}=H_{\mu,\Psi}-H_0$ for $\mu=p,n$.

Note that the scalars $\varepsilon_i$ can be seen as Lagrange multipliers; indeed, the nonlinear Dirac equations are the Euler-Lagrange equations of the energy functional
\begin{eqnarray}\label{eqenergyfunctional}
\mathcal{E}(\Psi)&=&\sum_{j=1}^A\int_{\mathbb{R}^3}\psi_j^\ast H_0\psi_j-\frac{g_\sigma^2}{8\pi}\int\int_{\mathbb{R}^3\times\mathbb{R}^3}\frac{\rho_s(x)\rho_s(y)}{|x-y|}e^{-m_\sigma|x-y|}\,dxdy\nonumber\\
&&+\frac{g_\omega^2}{8\pi}\int\int_{\mathbb{R}^3\times\mathbb{R}^3}\frac{\rho_0(x)\rho_0(y)}{|x-y|}e^{-m_\omega|x-y|}\,dxdy\nonumber\\
&&+\frac{g_\rho^2}{8\pi}\int\int_{\mathbb{R}^3\times\mathbb{R}^3}\frac{\rho_{00}(x)\rho_{00}(y)}{|x-y|}e^{-m_\rho|x-y|}\,dxdy\nonumber\\
&&+\frac{e^2}{8\pi}\int\int_{\mathbb{R}^3\times\mathbb{R}^3}\frac{\rho_{0}^c(x)\rho_{0}^c(y)}{|x-y|}\,dxdy
\end{eqnarray}
under the constraints $\int_{\mathbb{R}^3}\psi^\ast_i\psi_j=\delta_{ij}$ for $1\leq i,j\leq Z$ and for $Z+1\leq i,j \leq A$. Here we can suppose that the matrix of Lagrange multipliers is diagonal because of the fact that $\mathcal{E}(\Psi)$ is invariant under the transformations of $(\psi_1,\ldots,\psi_A)$ of the form 
\begin{equation*}
U=\left(\begin{array}{cc}U_p & 0\\ 0 & U_n\end{array}\right)
\end{equation*}
where $U_p$ (resp. $U_n$) is a $Z\times Z$ (resp. $N\times N$) unitary matrix. In the energy functional, we remark that only the $\sigma$ meson provides an attractive interaction.  Indeed, if $f$ is a real function,
$$
\int\int_{\mathbb{R}^3\times\mathbb{R}^3}\frac{f(x)f(y)}{|x-y|}e^{-\lambda|x-y|}\,dxdy=C\int_{\mathbb{R}^3}|\hat f(k)|^2\frac{1}{k^2+\lambda^2}\,dk
$$  
with $C$ a positive constant and $\hat f$ the Fourier transform of $f$. As a consequence, the term $$-\frac{g_\sigma^2}{8\pi}\int\int_{\mathbb{R}^3\times\mathbb{R}^3}\frac{\rho_s(x)\rho_s(y)}{|x-y|}e^{-m_\sigma|x-y|}\,dxdy$$ is negative and describes an attractive interaction.

Since the functional (\ref{eqenergyfunctional}) is not bounded from below under the constraints $\int_{\mathbb{R}^3}\psi^\ast_i\psi_j=\delta_{ij}$, as in \cite{estebansereNLDF} (see also \cite{estebanserelewin}), we introduce the following minimization problem
\begin{eqnarray}\label{eqminproblem}
I&=&\inf\left\{\mathcal{E}(\Psi);  \Psi\in (H^{1/2})^A, \int_{\mathbb{R}^3}\psi^\ast_i\psi_j=\delta_{ij}, 1\leq i,j\leq Z,  Z+1\leq i,j \leq A,\right.\nonumber\\
&&\left.\Lambda^-_{p,\Psi}(\psi_1,\ldots,\psi_Z)=0,\,\Lambda^-_{n,\Psi}(\psi_{Z+1},\ldots,\psi_A)=0\right\}
\end{eqnarray}
together with its extension
\begin{eqnarray}\label{eqminproblemlambda}
I\left(\lambda_1,\ldots,\lambda_A\right)&=&\inf\left\{\mathcal{E}(\Psi);  \Psi\in (H^{1/2})^A, \int_{\mathbb{R}^3}\psi^\ast_i\psi_j=\lambda_i\delta_{ij}, 1\leq i,j\leq Z,\right.\nonumber\\
&&\left. Z+1\leq i,j \leq A,\,\Lambda^-_{p,\Psi}(\psi_1,\ldots,\psi_Z)=0,\right.\nonumber\\
&&\left.\Lambda^-_{n,\Psi}(\psi_{Z+1},\ldots,\psi_A)=0\right\}
\end{eqnarray}
where, for $\mu=p,n$, $\Lambda^-_{\mu,\Psi}=\chi_{(-\infty,0)}({H}_{\mu,\Psi})$ is the negative spectral projector of the operator  ${H}_{\mu,\Psi}$, $$\Lambda^-_{p,\Psi}(\psi_1,\ldots,\psi_Z)=(\Lambda^-_{p,\Psi}\psi_1,\ldots,\Lambda^-_{p,\Psi}\psi_Z)=\Lambda^-_{p,\Psi}\Psi_p$$ and $$\Lambda^-_{n,\Psi}(\psi_{Z+1},\ldots,\psi_A)=(\Lambda^-_{n,\Psi}\psi_{Z+1},\ldots,\Lambda^-_{n,\Psi}\psi_A)=\Lambda^-_{n,\Psi}\Psi_n.$$

The idea of using a constraint of the form $\Lambda^-_{\mu,\Psi}\Psi_\mu=0$, for $\mu=p,n$, is due to M.J. Esteban and  E. Séré in the case of the Dirac-Fock equations (voir \cite{estebansereNLDF}). This constraint has a physical meaning; more precisely, if we neglect the vacuum polarization, the Dirac sea is represented by the negative spectral projector $\Lambda^-_{\mu,\Psi}$. Indeed, according to Dirac's original ideas, the vacuum is composed of infinitely many particles, which completely fill up the negative spectral subspace of ${H}_{\mu,\Psi}$: these particles form the Dirac sea. So, by 
Pauli exclusion principle,  the single-particle energies $\varepsilon_i$ should be strictly positive and, as a consequence, $\Psi_\mu$ should be in the positive spectral subspace of ${H}_{\mu,\Psi}$ for $\mu=p,n$. On the one hand, the use of the constraint $\Lambda^-_{\mu,\Psi}\Psi_\mu=0$ is very helpful since it transforms a strongly indefinite problem into a minimization problem; on the other hand, dealing with this constraint is the main difficulty of the proof of our results.

In this paper, we prove that, for $g_\sigma,g_\omega,g_\rho$ and $e$ sufficiently small, a solution of the equations (\ref{eqdiracyukawaproton}) and (\ref{eqdiracyukawaneutron}) can be obtained as a solution of the minimization problem (\ref{eqminproblem}).
\begin{thm}\label{thminsol}
If $g_\sigma,g_\omega,g_\rho$ and $e$ are sufficiently small, a minimizer of (\ref{eqminproblem}) is a solution of the equations (\ref{eqdiracyukawaproton}) and (\ref{eqdiracyukawaneutron}).
\end{thm}
Moreover, the application of the concentration-compactness method (\cite{lionscc1}, \cite{lionscc2}) to the minimization problem (\ref{eqminproblem}) yields the following theorem which is our main result.
\begin{thm}\label{thinequality} If $g_\sigma,g_\omega,g_\rho$ and $e$ are sufficiently small,
any minimizing sequence of (\ref{eqminproblem}) is relatively compact up to a translation if and only if the following condition holds
\begin{equation}\label{eqinequality}
I<I\left(\lambda_1,\ldots,\lambda_A\right)+I\left(1-\lambda_1,\ldots,1-\lambda_A\right)
\end{equation}
for all $\lambda_k \in [0,1]$, $k=1,\ldots,A$, such that $\sum\limits_{k=1}^A\lambda_k\in (0,A)$. \\
In particular, if (\ref{eqinequality}) holds, there exists a minimum of (\ref{eqminproblem}).
\end{thm}
This result is relevant  both from mathematical and physical point of view since it provides a condition that ensures the existence of a ground state solution of the equations (\ref{eqdiracyukawaproton}) and (\ref{eqdiracyukawaneutron}). Furthermore, this is the first result relating the existence of critical points of a strongly indefinite energy functional to strict con\-cen\-tra\-tion-compactness inequalities.

The condition $g_\sigma,g_\omega,g_\rho$ and $e$ sufficiently small means that we are in a weakly relativistic regime. In our proof of theorems \ref{thminsol} and \ref{thinequality}, this condition is required for several reasons. First of all, if  $g_\sigma,g_\omega,g_\rho$ and $e$ are sufficiently small, we can show that ${H}_{\mu,\Psi}$ is a self-adjoint isomorphism between $H^{1/2}$ and its dual $H^{-1/2}$, whose inverse is bounded independently of $\Psi$. 
Moreover, we need this condition to prove that a minimizing sequence of (\ref{eqminproblem}) is bounded in $\left(H^{1/2}(\mathbb{R}^3)\right)^A$. We remark that the estimates on $g_\sigma,g_\omega,g_\rho$ and $e$ are explicit up to this point. Finally, in both theorems, we have to apply the implicit function theorem with $g_\sigma,g_\omega,g_\rho$ and $e$ as parameters. 

This result is different from that obtained by Esteban--Séré on the Dirac--Fock equations (see \cite{estebansereSDF}, \cite{estebansereNLDF}). In  \cite{estebansereSDF}, by a more sophisticated variational method, Esteban--Séré found a infinite sequence of solutions of the Dirac-Fock equations and, in \cite{estebansereNLDF}, they showed that, in a weakly relativistic regime, the `` first '' solution of the Dirac--Fock equations found in \cite{estebansereSDF} can be viewed as an electronic ground state in the sense that it minimizes the Dirac--Fock energy among all electronic configurations which are orthogonal to the Dirac sea. Their variational method takes advantage of the fact that the Dirac--Fock energy functional is not translation invariant: it contains an attractive interaction term, due to the nucleus, which confines the electrons. The nonlinear interaction is rather purely repulsive so that the use of concentration-compactness is not necessary. On the contrary, the energy functional that we consider is  invariant under translations and one of the nonlinear interaction terms is attractive; because of the translation invariance, we are naturally led to use the concentration-compactness argument. 

In section \ref{secpropot}, we introduce some useful properties of the potential $V_{\mu,\Psi}$ and of the operator ${H}_{\mu,\Psi}$ for $\mu=p,n$. In section \ref{secproof}, we show how we can apply the concentration-compactness argument to the minimization problem (\ref{eqminproblem}). Finally, in section \ref{appsolutions}, we prove theorem \ref{thminsol}.

\section{Properties of the potential $V_{\mu,\Psi}$}\label{secpropot}
In this section, we describe some useful properties of the potential $V_{\mu,\Psi}$  and we give a condition on the parameters ($g_\sigma,g_\omega,g_\rho,e,N,Z$) which implies that $H_{\mu,\Psi}$ is a self-adjoint isomorphism and its inverse is bounded independently of $\Psi$.
\begin{lem}\label{lemregularitypot}For any $\Psi\in \left(H^{1/2}(\mathbb{R}^3)\right)^A$, 
$$
\begin{array}{ll}
V_{p,\Psi}\in L^{r}(\mathbb{R}^3), & 3<r< \infty\\
V_{n,\Psi}\in L^r(\mathbb{R}^3), & 1\le r< \infty.
\end{array}
$$ 
\end{lem}

\begin{proof}
The proof of this lemma is an application of Young's inequality : if $f\in L^p(\mathbb{R}^3)$, $g\in L^q(\mathbb{R}^3)$, then 
$$
\|f\star g\|_{L^r}\le \|f\|_{L^p}\|g\|_{L^q}
$$
with $1+\frac{1}{r}=\frac{1}{p}+\frac{1}{q}$, $1\le p,q,r\le\infty$.\\
We remark that if $\Psi\in \left(H^{1/2}(\mathbb{R}^3)\right)^A$, then $\rho_s$, $\rho_0$, $\rho_{00}$ and $\rho_0^c$ are in $L^p(\mathbb R^3)$ for $1\le p\le \frac{3}{2}$.\\
Furthermore, using the definition of the Gamma function, we can show that, for any  $ \lambda>0$, $\frac{e^{-\lambda|x|}}{|x|}\in L^q(\mathbb{R}^3)$ for $1\le q <3$.\\
Finally, we observe that $\frac{1}{|x|}$ can be written as $\frac{1}{|x|}=h_1(x)+h_2(x)$ with $h_1\in L^\alpha(\mathbb{R}^3)$ for $1\le \alpha < 3$ and $h_2\in  L^\beta(\mathbb{R}^3)$ for $3<\beta\le\infty$, where 
$h_1(x)=\frac{1}{|x|}$ for $|x|\le 1$, $h_1(x)=0$ otherwise.\\
Hence,
\begin{eqnarray*}
V_{ p,\Psi}&=&\underbrace{-\beta \frac{g_\sigma^2}{4\pi}\left(\frac{e^{-m_\sigma|\cdot|}}{|\cdot|}\star\rho_s\right)+\frac{g_\omega^2}{4\pi}\left(\frac{e^{-m_\omega|\cdot|}}{|\cdot|}\star\rho_0\right)+\frac{g_\rho^2}{4\pi}\left(\frac{e^{-m_\rho|\cdot|}}{|\cdot|}\star\rho_{00}\right)}\limits_{\in L^1(\mathbb{R}^3)\bigcap L^{r_c}(\mathbb{R}^3)}\\
&&+\underbrace{\frac{e^2}{4\pi}\left(h_1(x)\star\rho_0^c\right)}\limits_{\in L^1(\mathbb{R}^3)\bigcap L^{r_c}(\mathbb{R}^3)}+\underbrace{\frac{e^2}{4\pi}\left(h_2(x)\star\rho_0^c\right)}\limits_{\in L^r(\mathbb{R}^3),\ 3<r\le {r_c}}\in L^r(\mathbb{R}^3)
\end{eqnarray*}
for $3<r\le{r_c}$ with $r_c=\frac{9-3\varepsilon}{\varepsilon}$ for any $\varepsilon>0$, and 
\begin{equation*}
V_{ n,\Psi}=-\beta \frac{g_\sigma^2}{4\pi}\left(\frac{e^{-m_\sigma|\cdot|}}{|\cdot|}\star\rho_s\right)+\frac{g_\omega^2}{4\pi}\left(\frac{e^{-m_\omega|\cdot|}}{|\cdot|}\star\rho_0\right)-\frac{g_\rho^2}{4\pi}\left(\frac{e^{-m_\rho|\cdot|}}{|\cdot|}\star\rho_{00}\right){\in L^r(\mathbb{R}^3)}
\end{equation*} 
for $1\le r\le {r_c}$.\\
\end{proof}
For reader's convenience, let us remind the following lemma which lists some properties of $H_0$ and coulombic potential $V(x)=\frac{1}{|x|}$.
\begin{lem}[\cite{estebansereSDF}]\label{leminequalitycoulomb} The coulombic potential $V(x)=\frac{1}{|x|}$ satisfies the following Hardy-type inequalities:
\begin{equation}\label{eqcoulomb1}
\left(\varphi,(\mu\star V)\varphi\right)_{L^2}\le \frac{1}{2}\left(\frac{\pi}{2}+\frac{2}{\pi}\right)\left(\varphi,|H_0|\varphi\right)_{L^2},
\end{equation}
for all $\varphi\in \Lambda^+(H^{1/2})\cup\Lambda^-(H^{1/2})$ and for all probability measures $\mu$ on $\mathbb R^3$. Moreover,
\begin{equation}\label{eqcoulomb2}
\left(\varphi,(\mu\star V)\varphi\right)_{L^2}\le \frac{\pi}{2}\left(\varphi,|H_0|\varphi\right)_{L^2},\ \forall \varphi \in H^{1/2},
\end{equation}
\begin{equation}\label{eqcoulomb3}
\left\|(\mu\star V)\varphi\right\|_{L^2}\le 2\left\|\nabla\varphi\right\|_{L^2},\ \forall \varphi \in H^{1}.
\end{equation}
\end{lem}
In the particular case where $\mu$ is equal to the Dirac mass at the origin $\delta_0$, we refer to Burenkov--Evans (\cite{burenkovevans}) and Tix (\cite{tixpos}, \cite{tixlow}) for the inequality (\ref{eqcoulomb1}), to Herbst (\cite{herbst}) and Kato (\cite{kato}) for (\ref{eqcoulomb2}) and to Thaller's book (\cite{thaller}) for the standard Hardy inequality (\ref{eqcoulomb3}). The extension of (\ref{eqcoulomb1}), (\ref{eqcoulomb2}) and (\ref{eqcoulomb3}) to a general probability measure $\mu$ is immediate.

Then, using lemma \ref{leminequalitycoulomb} and proceeding like in \cite{estebansereSDF} (Lemma 3.1), we obtain the following estimates.
\begin{lem}\label{lemestimate}
Assume that 
\begin{eqnarray}\label{eqcondg1}
\frac{g_\sigma^2A+g_\rho^2\max(Z,N)}{4\pi}&<&\frac{2}{\pi/2+2/\pi},\\
\label{eqcondg2}
\frac{g_\sigma^2A+g_\omega^2A+g_\rho^2Z+e^2Z}{4\pi}&<&\frac{2}{\pi/2+2/\pi},\\
\label{eqcondg3}
\frac{g_\sigma^2A+g_\omega^2A+g_\rho^2N}{4\pi}&<&\frac{2}{\pi/2+2/\pi}.
\end{eqnarray}
There is a constant $h_{\mu}>0$, such that for any $\Psi\in  \left(H^{1/2}(\mathbb{R}^3)\right)^A$ such that $$\mathrm{Gram}_{L^2}(\Psi)\le \mathbbm{1},$$ and $\psi\in  H^{1/2}(\mathbb{R}^3)$,
\begin{equation}\label{eqestimatepot}
h_{\mu}\|\psi\|_{H^{1/2}}\le\|{H}_{\mu,\Psi}\psi\|_{H^{-1/2}}
\end{equation} 
with $\mu=p,n$. In other words, ${H}_{\mu,\Psi}$ is a self-adjoint isomorphism between $H^{1/2}$ and its dual $H^{-1/2}$, whose inverse is bounded independently of $\Psi$.
\end{lem}
Finally, a straightforward application of the inequality (\ref{eqcoulomb3}) yields the following lemma.
\begin{lem}\label{lemestimate2}
Assume that
\begin{eqnarray}\label{eqcondg4}
d_p&=&\frac{(g^2_\sigma+g^2_\omega+g^2_\rho)A+e^2Z}{2\pi}<1,\\
\label{eqcondg5}
d_n&=&\frac{(g^2_\sigma+g^2_\omega+g^2_\rho)A}{2\pi}<1.
\end{eqnarray}
For any $\Psi\in  \left(H^{1/2}(\mathbb{R}^3)\right)^A$ such that $\mathrm{Gram}_{L^2}(\Psi)\le \mathbbm{1}$,
\begin{eqnarray}\label{eqestimatepot2}
V_{\mu,\Psi}&\le& d_{\mu}^{1/2}|H_0|\\
\label{eqestimatepot3}
(1-d_\mu)^{1/2}|H_0|&\le& |H_{\mu,\Psi}|
\end{eqnarray}
for $\mu=p,n$.
\end{lem}

\begin{rem}
Our estimates are far from optimal. In particular, we do not give any condition on $m_\sigma$, $m_\omega$ and $m_\rho$. We can expect that taking into account the meson masses, one can obtain better estimates.  
\end{rem}

\section{Proof of theorem \ref{thinequality}}\label{secproof}

This theorem is an application of the concentration-compactness argument (see \cite{lionscc1}, \cite{lionscc2}). Like in \cite{gognylions}, if $(\psi_1^k,\ldots,\psi_A^k)$ is a minimizing sequence of (\ref{eqminproblem}), then we apply the lemma below (proved in \cite{lionscc1}) with the probability $P_k$ in $\mathbb{R}^3$ whose density is $\frac{1}{A}\rho^k$ and $\rho^k=\sum\limits_{i=1}^A|\psi^k_i|^2$.
\begin{lem}\label{lemcc}
Let $(P_k)_k$ be a sequence of probability measures on $\mathbb{R}^N$. Then there exists a subsequence that we still denote by $P_k$ such that one of the following properties holds:
\begin{enumerate}
\item (compactness up to a translation) $\exists y^k\in \mathbb{R}^N$, $\forall \varepsilon>0$, $\exists R<\infty$
$$
P_k\left(B\left(y^k,R\right)\right)\geq1-\varepsilon;
$$
\item (vanishing) $\forall R<\infty$
$$
\sup\limits_{y\in \mathbb{R}^N}P_k\left(B\left(y,R\right)\right) \xrightarrow[k]{} 0;
$$
\item (dichotomy) $\exists \alpha \in(0,1)$, $\forall \varepsilon >0$, $\forall M < \infty$, $\exists R_0 \ge M$, $\exists y^k\in \mathbb{R}^N$, $\exists R_k \xrightarrow[k]{} +\infty$ such that 
$$
\begin{array}{ll}
\left|P_k\left(B\left(y^k,R_0\right)\right)-\alpha\right|\le \varepsilon, & \left|P_k\left(B\left(y^k,R_k\right)^c\right)-(1-\alpha)\right|\le \varepsilon.
\end{array}
$$
\end{enumerate}
\end{lem}

In the following subsections, we prove that if the condition (\ref{eqinequality}) holds, then we can rule out dichotomy and vanishing. 

First, we make a few preliminary observations; let $\Psi^k=(\psi_1^k,\ldots,\psi_A^k)$ be  a minimizing sequence and $g_\sigma,g_\omega, g_\rho$ and $e$ such that $d_\mu<\frac{4}{5}$ for $\mu=p,n$, then $\Psi^k$ is bounded in $\left(H^{1/2}(\mathbb{R}^3)\right)^A$. Indeed, since $\Psi^k$ is a minimizing sequence, there exists a constant $C$ such that 
\begin{eqnarray*}
C\ge \mathcal{E}(\Psi^k)&=&\sum_{j=1}^A\left(\psi_j^k, H_0\psi_j^k\right)_{L^2}+\frac{1}{2}\sum_{j=1}^Z\left(\psi_j^k, V_{p,\Psi^k}\psi_j^k\right)_{L^2}\\
&&+\frac{1}{2}\sum_{j=Z+1}^A\left(\psi_j^k, V_{n,\Psi^k}\psi_j^k\right)_{L^2}\\
&=&\sum_{j=1}^Z\left(\psi_j^k, H_{p,\Psi^k}\psi_j^k\right)_{L^2}+\sum_{j=Z+1}^A\left(\psi_j^k, H_{n,\Psi^k}\psi_j^k\right)_{L^2}\\
&&-\frac{1}{2}\sum_{j=1}^Z\left(\psi_j^k, V_{p,\Psi^k}\psi_j^k\right)_{L^2}-\frac{1}{2}\sum_{j=Z+1}^A\left(\psi_j^k, V_{n,\Psi^k}\psi_j^k\right)_{L^2}\\
\end{eqnarray*}
Then, using the fact that, for any $k\in \mathbb N$, $\psi_j^k=\Lambda^+_{p,\Psi^k}\psi_j^k$ for $1\le j\le Z$, $\psi_j^k=\Lambda^+_{n,\Psi^k}\psi_j^k$ for $Z+1\le j\le A$  and the inequalities (\ref{eqestimatepot2}) and (\ref{eqestimatepot3}), we obtain 
\begin{align*}
C\ge&\sum_{j=1}^Z\left(\psi_j^k, |H_{p,\Psi^k}|\psi_j^k\right)_{L^2}+\sum_{j=Z+1}^A\left(\psi_j^k, |H_{n,\Psi^k}|\psi_j^k\right)_{L^2}\\
&-\frac{1}{2}\sum_{j=1}^Z d_p^{1/2}\left(\psi_j^k, |H_0|\psi_j^k\right)_{L^2}-\frac{1}{2}\sum_{j=Z+1}^A d_n^{1/2}\left(\psi_j^k, |H_0|\psi_j^k\right)_{L^2}\\
\ge&\sum_{j=1}^Z\left[(1-d_p)^{1/2}-\frac{d_p^{1/2}}{2}\right]\|\psi_j^k\|^2_{H^{1/2}}+\sum_{j=Z+1}^A\left[(1-d_n)^{1/2}-\frac{d_n^{1/2}}{2}\right]\|\psi_j^k\|^2_{H^{1/2}}.
\end{align*}
As a conclusion, if $2(1-d_\mu)^{1/2}-d_\mu^{1/2}>0$, that means $d_\mu< \frac{4}{5}$, then $\|\Psi^k\|^2_{(H^{1/2})^A}$ is bounded independently of $k$ and $I$ is bounded from below.

\subsection{Dichotomy does not occur}

If dichotomy occurs (case iii.), then, roughly speaking, $\Psi^k=(\psi_1^k,\ldots,\psi_A^k)$ can be split into two parts that we denote by $\Psi_1^k=(\psi_{1,1}^k,\ldots,\psi_{A,1}^k)$ and $\Psi_2^k=(\psi_{1,2}^k,\ldots,\psi_{A,2}^k)$. 
More precisely, let $\xi,\zeta$ be cut-off functions: $0\le\xi\le 1$, $0\le\zeta\le1$, $\xi(x)=1$ if $|x|\le 1$, $\xi(x)=0$ if $|x|\ge 2$, $\zeta(x)=0$ if $|x|\le 1$, $\zeta(x)=1$ if $|x|\ge 2$, $\xi,\zeta\in \mathcal D (\mathbb R^3)$ and let $\xi_\mu$, $\zeta_\mu$ denote $\xi\left(\frac{\cdot}{\mu}\right)$, $\zeta\left(\frac{\cdot}{\mu}\right)$. We set
\begin{eqnarray*}
\psi_{i,1}^k&=&\xi_{\frac{R_k}{8}}(\cdot-y^k)\psi_{i}^k\\
\psi_{i,2}^k&=&\zeta_{\frac{R_k}{2}}(\cdot-y^k)\psi_{i}^k
\end{eqnarray*}
with $R_k \xrightarrow[k]{} +\infty$. We remind that $\mathrm{dist}\left(\mathrm{supp}\ \psi_{i,1}^k, \mathrm{supp}\ \psi_{i,2}^k  \right)>0$,
\begin{equation}\label{eqdichotomydefconv}
\left\|\psi_{i}^k-\left(\psi_{i,1}^k+\psi_{i,2}^k\right)\right\|_{L^p}\xrightarrow[k]{}0
\end{equation}
for $2\le p< 3$, and
\begin{equation}\label{eqdichotomydefconvH}
\left\|\psi_{i}^k-\left(\psi_{i,1}^k+\psi_{i,2}^k\right)\right\|_{H^{1/2}}\xrightarrow[k]{}0
\end{equation}
(see \cite{lionscc1}, \cite{lenzmannlewinboson}).

Next, we may assume that 
\begin{equation}\label{eqgramdichotomy}
\begin{array}{ll}
\int_{\mathbb{R}^3}\psi^{k^\ast}_{i,1}\psi^k_{j,1}=\lambda_i\delta_{ij}, & \int_{\mathbb{R}^3}\psi^{k^\ast}_{i,2}\psi^k_{j,2}=(1-\lambda_i)\delta_{ij},
\end{array}
\end{equation}
for $1\leq i,j\leq Z$, $Z+1\leq i,j \leq A$, with $\lambda_i\in [0,1]$ such that $\sum\limits_{i=1}^A\lambda_i\in (0,A)$.
In fact, suppose that $\Theta^k=(\theta_1^k,\ldots,\theta_A^k)$ is a minimizing sequence for which the dichotomy case occurs. We remind that $\mathcal{E}(\Psi)$ is invariant under the transformations of $(\psi_1,\ldots,\psi_A)$ of the form 
\begin{equation*}\label{eqtransinv}
U=\left(\begin{array}{cc}U_p & 0\\ 0 & U_n\end{array}\right)
\end{equation*}
where $U_p$ (resp. $U_n$) is a $Z\times Z$ (resp. $N\times N$) unitary matrix; then, using this kind of transformations and writing $\Theta_1^k=(\Theta_{p,1}^k,\Theta_{n,1}^k)$, it is clear that we may diagonalize $\mathrm{Gram}_{L^2}(\Theta_{p,1}^k)$ and $\mathrm{Gram}_{L^2}(\Theta_{p,1}^k)$. In particular, we have
\begin{align*}
&\mathrm{Gram}_{L^2}(\Theta_{p,1}^k)=\mathrm{diag}(\lambda_1^k,\ldots,\lambda_Z^k),\\ &\mathrm{Gram}_{L^2}(\Theta_{n,1}^k)=\mathrm{diag}(\lambda_{Z+1}^k,\ldots,\lambda_A^k)
\end{align*}
with $0\le \lambda_i^k\le 1$ for all $k\in \mathbb{N}$. 
Since $\{\lambda_i^k\}_k$ is a bounded sequence in $\mathbb{R}$, up to a subsequence $\lambda_i^k\xrightarrow[k]{}\lambda_i$ with $0\le \lambda_i\le 1$, and 
\begin{equation}\label{eqconvgramdichotomy1}
\begin{aligned}
\mathrm{Gram}_{L^2}(\Theta_{p,1}^k)-\Delta_{p,1}\xrightarrow[k]{} 0 &\qquad \Delta_{p,1}=\mathrm{diag}(\lambda_1,\ldots,\lambda_Z),\\
\mathrm{Gram}_{L^2}(\Theta_{n,1}^k)-\Delta_{n,1}\xrightarrow[k]{} 0 &\qquad \Delta_{n,1}=\mathrm{diag}(\lambda_{Z+1},\ldots,\lambda_A).
\end{aligned}
\end{equation}
Moreover, as a consequence of the definition of $\Theta^k_1$ and $\Theta^k_2$, we have
\begin{equation}\label{eqconvgramdichotomy2}
\begin{aligned}
\mathrm{Gram}_{L^2}(\Theta_{p,2}^k)-\Delta_{p,2}\xrightarrow[k]{} 0 &\qquad \Delta_{p,2}=\mathrm{diag}(1-\lambda_1,\ldots,1-\lambda_Z),\\ 
\mathrm{Gram}_{L^2}(\Theta_{n,2}^k)-\Delta_{n,2}\xrightarrow[k]{} 0 &\qquad \Delta_{n,2}=\mathrm{diag}(1-\lambda_{Z+1},\ldots,1-\lambda_A).
\end{aligned}
\end{equation}
So, to obtain (\ref{eqgramdichotomy}), we proceed as follows. 

First of all, we define the sets $\mathcal{I}_p=\{i\in\mathbb{N}:1\le i \le Z\}$, $\mathcal{I}_{p,1}=\{i\in \mathcal{I}_p : \lambda_i =0\}$, $\mathcal{I}_{p,2}=\{i\in \mathcal{I}_p : \lambda_i =1\}$, $\tilde{\mathcal{I}}_{p,1}= \mathcal{I}_p\smallsetminus \mathcal{I}_{p,1}$ and $\tilde{\mathcal{I}}_{p,2}= \mathcal{I}_p\smallsetminus \mathcal{I}_{p,2}$.

Second, if $i\in \mathcal{I}_{p,1}$, we replace $\theta_{i,1}^k$ with $\psi_{i,1}^k=0$ and, in the same way, if $i\in \mathcal{I}_{p,2}$, we replace $\theta_{i,2}^k$ with $\psi_{i,2}^k=0$.

Next, we denote
\begin{align*}
\tilde\Theta^{k}_{p,1}=\left(\theta^k_{i,1}\right)_{i\in \tilde{\mathcal{I}}_{p,1}}, & \qquad \tilde G^k_{p,1}= \mathrm{Gram}_{L^2}(\tilde\Theta_{p,1}^k),\\ 
\tilde\Theta^{k}_{p,2}=\left(\theta^k_{i,2}\right)_{i\in \tilde{\mathcal{I}}_{p,2}}, & \qquad \tilde G^k_{p,2}= \mathrm{Gram}_{L^2}(\tilde\Theta_{p,2}^k),\\ 
\tilde{\Delta}_{p,1}=\mathrm{diag}\left(\lambda_i\right)_{i\in \tilde{\mathcal{I}}_{p,1}}, & \qquad \tilde{\Delta}_{p,2}=\mathrm{diag}\left(1-\lambda_i\right)_{i\in \tilde{\mathcal{I}}_{p,2}},
\end{align*}
and we remark that  $\tilde{\Delta}_{p,1}$ and  $\tilde{\Delta}_{p,2}$ are invertible matrices. Furthermore, using (\ref{eqconvgramdichotomy1}) and (\ref{eqconvgramdichotomy2}), we can write, for $k\to+\infty$,
\begin{align*}
\tilde G^k_{p,1}=\tilde{\Delta}_{p,1} + o(1), &\qquad \tilde G^k_{p,2}=\tilde{\Delta}_{p,2} + o(1);
\end{align*}
then $\tilde G^k_{p,1}$ and $\tilde G^k_{p,2}$ are invertible matrices and
\begin{align*}
(\tilde G^k_{p,1})^{-1/2}=\tilde{\Delta}_{p,1}^{-1/2} + o(1), &\qquad (\tilde G^k_{p,2})^{-1/2}=\tilde{\Delta}_{p,2}^{-1/2} + o(1).
\end{align*}

To conclude, it is enough to consider a small perturbation of $\tilde\Theta_{p,1}^k$ and $\tilde\Theta_{p,2}^k$. More precisely, we take 
\begin{align*}
\tilde\Psi^k_{p,1}&=\left(\psi^k_{i,1}\right)_{i\in \tilde{\mathcal{I}}_{p,1}}=\tilde\Theta^{k}_{p,1}(\tilde G^k_{p,1})^{-1/2}\tilde{\Delta}_{p,1}^{1/2}\\
\tilde\Psi^k_{p,2}&=\left(\psi^k_{i,2}\right)_{i\in \tilde{\mathcal{I}}_{p,2}}=\tilde\Theta^{k}_{p,2}(\tilde G^k_{p,2})^{-1/2}\tilde{\Delta}_{p,2}^{1/2}.
\end{align*}
By a straightforward calculation, we obtain $$\mathrm{Gram}_{L^2}(\tilde\Psi_{p,1}^k)= \tilde{\Delta}_{p,1}\  \mbox{and}\ \mathrm{Gram}_{L^2}(\tilde\Psi_{p,2}^k)= \tilde{\Delta}_{p,2};$$ hence, writing $\Psi^k_{p,1}=(\psi^k_{1,1},\ldots,\psi^k_{Z,1})$ and $\Psi^k_{p,2}=(\psi^k_{1,2},\ldots,\psi^k_{Z,2})$, we have 
$$
\mathrm{Gram}_{L^2}(\Psi_{p,1}^k)=\mathrm{diag}(\lambda_1,\ldots,\lambda_Z)\ \mbox{and}\ \mathrm{Gram}_{L^2}(\Psi_{p,2}^k)=\mathrm{diag}(1-\lambda_1,\ldots,1-\lambda_Z).
$$

Finally, with the same arguments, we can construct $\Psi^k_{n,1}=(\psi^k_{Z+1,1},\ldots,\psi^k_{A,1})$ and $\Psi^k_{n,2}=(\psi^k_{Z+1,2},\ldots,\psi^k_{A,2})$ such that
$
\mathrm{Gram}_{L^2}(\Psi_{n,1}^k)=\mathrm{diag}(\lambda_{Z+1},\ldots,\lambda_A)$ and  $\mathrm{Gram}_{L^2}(\Psi_{n,2}^k)=\mathrm{diag}(1-\lambda_{Z+1},\ldots,1-\lambda_A)$.

We remark that $\Psi_{1}^k=(\psi_{1,1}^k,\ldots,\psi_{A,1}^k)$ and $\Psi_{2}^k=(\psi_{1,2}^k,\ldots,\psi_{A,2}^k)$ do not necessarily satisfy the constraints of $I\left(\lambda_1,\ldots,\lambda_A\right)$ and $I\left(1-\lambda_1,\ldots,1-\lambda_A\right)$ respectively, then we proceed as follows. \\
First, we show that
\begin{equation}\label{eqcondproj1}
\begin{array}{l}
\Lambda^-_{p,\Psi_1^k}(\psi^k_{1,1},\ldots,\psi^k_{Z,1})\xrightarrow[k]{}0\ \mbox{in}\ \left(H^{1/2}(\mathbb R^3)\right)^Z,\\[6pt] \Lambda^-_{n,\Psi_1^k}(\psi^k_{Z+1,1},\ldots,\psi^k_{A,1})\xrightarrow[k]{}0\ \mbox{in}\ \left(H^{1/2}(\mathbb R^3)\right)^N,
\end{array}
\end{equation}
and
\begin{equation}\label{eqcondproj2}
\begin{array}{l}
\Lambda^-_{p,\Psi_2^k}(\psi^k_{1,2},\ldots,\psi^k_{Z,2})\xrightarrow[k]{}0\ \mbox{in}\ \left(H^{1/2}(\mathbb R^3)\right)^Z, \\[6pt] \Lambda^-_{n,\Psi_2^k}(\psi^k_{Z+1,2},\ldots,\psi^k_{A,2})\xrightarrow[k]{}0\ \mbox{in}\ \left(H^{1/2}(\mathbb R^3)\right)^N.
\end{array}
\end{equation}
Second, using the implicit function theorem, we construct $$\Phi_1^k=(\Phi_{p,1}^k,\Phi_{n,1}^k), \Phi_2^k=(\Phi_{p,2}^k,\Phi_{n,2}^k)\in \left(H^{1/2}(\mathbb R^3)\right)^Z\times\left(H^{1/2}(\mathbb R^3)\right)^N,$$ small perturbations of $\Psi_1^k$, $\Psi_2^k$ in $\left(H^{1/2}(\mathbb R^3)\right)^A$, such that
\begin{equation}\label{eqcondproj1Phi}
\begin{array}{ll}
\Lambda^-_{p,\Phi_1^k}\Phi^k_{p,1}=0, & \Lambda^-_{n,\Phi_1^k}\Phi^k_{n,1}=0,
\end{array}
\end{equation}
\begin{equation}\label{eqcondproj2Phi}
\begin{array}{ll}
\Lambda^-_{p,\Phi_2^k}\Phi^k_{p,2}=0, & \Lambda^-_{n,\Phi_2^k}\Phi^k_{n,2}=0
\end{array}
\end{equation}
and
\begin{equation}\label{eqcondGram}
\mathrm{Gram}_{L^2}(\Phi_{\mu,i}^k)=\mathrm{Gram}_{L^2}(\Psi_{\mu,i}^k)
\end{equation}
for $\mu=p,n$ and $i=1,2$.\\
In conclusion, thanks to the continuity of $\mathcal{E}$ in $H^{1/2}(\mathbb R^3)$, we obtain
\begin{eqnarray*}
I&=&\lim_{k\to\infty}\mathcal{E}(\Psi^k)\ge \varliminf_{k\to\infty}\mathcal{E}(\Psi_1^k)+\varliminf_{k\to\infty}\mathcal{E}(\Psi_2^k)\nonumber\\
&=& \varliminf_{k\to\infty}\mathcal{E}(\Phi_1^k)+\varliminf_{k\to\infty}\mathcal{E}(\Phi_2^k)\nonumber\\
&\ge&I\left(\lambda_1,\ldots,\lambda_A\right)+I\left(1-\lambda_1,\ldots,1-\lambda_A\right)
\end{eqnarray*}
that clearly contradicts (\ref{eqinequality}). We remind that the first inequality is obtained by using the properties of localization of $\Psi_1^k,\Psi_2^k,\nabla\Psi_1^k$ and $\nabla\Psi_2^k$.

We start by showing that 
$$
\Lambda^-_{p,\Psi_1^k}(\psi^k_{1,1},\ldots,\psi^k_{Z,1})\xrightarrow[k]{}0\ \mbox{in}\ \left(H^{1/2}(\mathbb{R}^3)\right)^Z.
$$
Using the formula (see \cite{kato})
\begin{eqnarray}\label{eqKato}
\Lambda^-_B-\Lambda^-_A&=&\frac{1}{2\pi}\int_{-\infty}^{+\infty}\left[(A-i\eta)^{-1}-(B-i\eta)^{-1}\right]\,d\eta\nonumber\\
&=&\frac{1}{2\pi}\int_{-\infty}^{+\infty}(A-i\eta)^{-1}(B-A)(B-i\eta)^{-1}\,d\eta,
\end{eqnarray}
we can write 
\begin{eqnarray}\label{eqdiff}
&&\Lambda^-_{p,\Psi_1^k}\psi^k_{i,1}-\Lambda^-_{p,\Psi^k}\psi^k_{i,1}=\nonumber\\
&&=\frac{1}{2\pi}\int_{-\infty}^{+\infty}({H}_{p,\Psi^k}-i\eta)^{-1}({H}_{p,\Psi_1^k}-{H}_{p,\Psi^k})({H}_{p,\Psi_1^k}-i\eta)^{-1}\psi^k_{i,1}\,d\eta.
\end{eqnarray}
for $i=1,\ldots,Z$. Hence, if we prove that 
$$
\left\|\Lambda^-_{p,\Psi^k}\psi^k_{i,1}\right\|_{H^{1/2}}\xrightarrow[k]{}0
$$
and
$$
\int_{-\infty}^{+\infty}\left\|({H}_{p,\Psi^k}-i\eta)^{-1}({H}_{p,\Psi_1^k}-{H}_{p,\Psi^k})({H}_{p,\Psi_1^k}-i\eta)^{-1}\psi^k_{i,1}\right\|_{H^{1/2}}\,d\eta\xrightarrow[k]{}0,
$$
we can conclude that 
$$
\left\|\Lambda^-_{p,\Psi_1^k}\psi^k_{i,1}\right\|_{H^{1/2}}\xrightarrow[k]{}0.
$$

First of all, we consider $$f^k(\eta)=\left\|({H}_{p,\Psi^k}-i\eta)^{-1}({H}_{p,\Psi_1^k}-{H}_{p,\Psi^k})({H}_{p,\Psi_1^k}-i\eta)^{-1}\psi^k_{i,1}\right\|_{H^{1/2}}$$ and we prove that, $\forall \eta \in \mathbb R$, 
$$
f^k(\eta)\xrightarrow[k]{}0. 
$$
We decompose the proof of this fact into two lemmas.
\begin{lem}\label{lemdomination} Assume that $g_\sigma,g_\omega,g_\rho$ and $e$ are sufficiently small and let $\Psi^k$ be a sequence in $\left(H^{1/2}(\mathbb R^3)\right)^A$ such that $\mathrm{Gram}_{L^2}(\Psi^k)\le \mathbbm{1}$, $\forall k\in \mathbb N$, and $\left\|\Psi^k\right\|_{(L^p)^A}$ is bounded independently of $k$ for $2\le p\le 3$.
Then for any $\varphi\in L^2(\mathbb{R}^3)$ and for any $\eta \in \mathbb{R}$, there exists a constant $\hat h_p$ such that
\begin{equation}\label{eqnormH12}
\left\|({H}_{p,\Psi^k}-i\eta)^{-1}\varphi\right\|_{H^{1/2}}\le \frac{1}{(m_b^2+\eta^2)^{1/4}}\left(\|\varphi\|_{L^2}+\frac{C}{(\hat h_p^2+\eta^2)^{1/3}}\left\|{\varphi}\right\|_{L^2}\right)
\end{equation}
with $C$ a constant that does not depend on $k$.
\end{lem}
\begin{proof}
First of all, we write
\begin{eqnarray*}
({H}_{p,\Psi^k}-i\eta)^{-1}\varphi&=&\chi\\
\varphi&=&(H_{p,\Psi^k}-i\eta)\chi\\
\varphi&=&(H_{0}-i\eta)\chi+V_{p,\Psi^k}\chi\\
(H_{0}-i\eta)^{-1}(\varphi-V_{p,\Psi^k}\chi)&=&\chi\\
\end{eqnarray*}
where $V_{p,\Psi^k}=H_{p,\Psi^k}-H_0$. It is easy to show that if $\varphi\in L^2(\mathbb{R}^3)$, then $\chi \in H^{1/2}(\mathbb{R}^3)$; indeed, there exists a constant $h_p>0$ such that
\begin{eqnarray*}
\left\|\varphi\right\|^2_{L^2}&=&\left((H_{p,\Psi^k}-i\eta)\chi,(H_{p,\Psi^k}-i\eta)\chi\right)_{L^2}=\left\|H_{p,\Psi^k}\chi\right\|^2_{L^2}+\eta^2\left\|\chi\right\|^2_{L^2}\\
&&\ge m_bh_p^2\left\|\chi\right\|^2_{H^{1/2}}+\eta^2\left\|\chi\right\|^2_{L^2}\ge m_b^2h_p^2\left\|\chi\right\|^2_{L^{2}}+\eta^2\left\|\chi\right\|^2_{L^2}\\
\end{eqnarray*}
thanks to Sobolev embeddings and lemma \ref{lemestimate}.\\
Next, to have a good estimate of the $H^{1/2}$-norm, we use its definition and we obtain
\begin{eqnarray*}
\|\chi\|^2_{H^{1/2}}&=&\left\|(H_{0}-i\eta)^{-1}(\varphi-V_{p,\Psi^k}\chi)\right\|^2_{H^{1/2}}\\
&=&\int_{\mathbb{R}^3}(m_b^2+|p|^2)^{1/2}\left|\frac{\hat{\varphi}(p)-\widehat{V_{p,\Psi^k}\chi}(p)}{\hat{H_0}(p)-i\eta}\right|^2\,dp\\
&\le&\int_{\mathbb{R}^3}\frac{(m_b^2+|p|^2)^{1/2}}{\hat{H_0}(p)^2+\eta^2}\left(\left|\hat{\varphi}(p)\right|^2+\left|\widehat{V_{p,\Psi^k}\chi}(p)\right|^2\right)\,dp\\
&=&\int_{\mathbb{R}^3}\frac{(m_b^2+|p|^2)^{1/2}}{m_b^2+|p|^2+\eta^2}\left(\left|\hat{\varphi}(p)\right|^2+\left|\widehat{V_{p,\Psi^k}\chi}(p)\right|^2\right)\,dp\\
&\le&\int_{\mathbb{R}^3}\frac{(m_b^2+|p|^2)^{1/2}}{(m_b^2+|p|^2)^{1/2}(m_b^2+\eta^2)^{1/2}}\left(\left|\hat{\varphi}(p)\right|^2+\left|\widehat{V_{p,\Psi^k}\chi}(p)\right|^2\right)\,dp\\
&\le&\frac{1}{(m_b^2+\eta^2)^{1/2}}\left(\left\|{\varphi}\right\|^2_{L^2 }+\left\|{V_{p,\Psi^k}\chi}\right\|^2_{L^2}\right).
\end{eqnarray*}
To conclude, we have to find an estimate for $\left\|{V_{p,\Psi^k}\chi}\right\|_{L^2}$. In particular, we have
\begin{eqnarray*}
\left\|{V_{p,\Psi^k}\chi}\right\|_{L^2}&\le&\left\|{V_{p,\Psi^k}}\right\|_{L^{18}}\left\|{\chi}\right\|_{L^{9/4}}\le\left\|{V_{p,\Psi^k}}\right\|_{L^{18}} \left\|{\chi}\right\|_{L^{2}}^{2/3}\left\|{\chi}\right\|_{L^{3}}^{1/3} \\
&\le&{C}\left\|{\varphi}\right\|_{L^2}^{1/3}\left\|(H_{p,\Psi^k}-i\eta)^{-1}\varphi\right\|^{2/3}_{L^2}\le\frac{C}{(\hat h_p^2+\eta^2)^{1/3}}\left\|{\varphi}\right\|_{L^2}\\
\end{eqnarray*}
where $\hat h_p= m_b h_p$ and $C$ is a constant that does not depend on $k$. Hence, 
\begin{equation*}
\left\|(H_{p,\Psi^k}-i\eta)^{-1}\varphi\right\|_{H^{1/2}}\le \frac{1}{(m_b^2+\eta^2)^{1/4}}\left(\|\varphi\|_{L^2}+\frac{ C}{(\hat h_p^2+\eta^2)^{1/3}}\left\|{\varphi}\right\|_{L^2}\right)
\end{equation*}
$\forall \eta\in \mathbb{R}$.\\
\end{proof}
\begin{lem}\label{lemproj} Assume that $g_\sigma,g_\omega,g_\rho$ and $e$ are sufficiently small; let $\Psi^k$ be a minimizing sequence of (\ref{eqminproblem}) and $\Psi_1^k$, $\Psi_2^k$ defined as above. Then, for $i=1,\ldots,Z$,
\begin{equation}
\begin{array}{ll}
(H_{p,\Psi_1^k}-H_{p,\Psi^k})(H_{p,\Psi_1^k}-i\eta)^{-1}\psi^k_{i,1}\xrightarrow[k]{L^2}0 & \forall \eta \in \mathbb{R}.
\end{array}
\end{equation}
\end{lem}
\begin{proof}
First of all, we study the behavior of 
$$
(H_{p,\Psi_1^k}-i\eta)^{-1}\psi^k_{i,1}
$$
in $\mathbb{R}^3\smallsetminus B\left(y^k,\frac{R_k}{4}\right)$. Writing $\xi_{\frac{R_k}{8},y^k}(\cdot)=\xi_{\frac{R_k}{8}}(\cdot-y^k)$, we obtain
\begin{align}\label{eqcommpsi1}
(H_{p,\Psi_1^k}-i\eta)^{-1}\psi^k_{i,1}=(H_{p,\Psi_1^k}-i\eta)^{-1}\xi_{\frac{R_k}{8},y^k}\psi^k_{i}=\xi_{\frac{R_k}{8},y^k}(H_{p,\Psi_1^k}-i\eta)^{-1}\psi^k_{i}+\tau^k
\end{align}
with $\tau^k$ defined by
\begin{equation}\label{eqdeftaupsi1}
\tau^k=\left[(H_{p,\Psi_1^k}-i\eta)^{-1},\xi_{\frac{R_k}{8},y^k}\right]\psi^k_{i}.
\end{equation}
As a consequence, we have
\begin{align}\label{eqnormpsi1ball}
\left\|(H_{p,\Psi_1^k}-i\eta)^{-1}\psi^k_{i,1}\right\|_{L^2{\scriptscriptstyle\left(\mathbb{R}^3\smallsetminus B\left(y^k,\frac{R_k}{4}\right)\right)}}&\le\left\|\xi_{\frac{R_k}{8},y^k}(H_{p,\Psi_1^k}-i\eta)^{-1}\psi^k_{i}\right\|_{L^2{\scriptscriptstyle\left(\mathbb{R}^3\smallsetminus B\left(y^k,\frac{R_k}{4}\right)\right)}}\nonumber\\
&+\left\|\tau^k\right\|_{L^2(\mathbb{R}^3)}\nonumber\\
&\le C \left\|\xi_{\frac{R_k}{8},y^k}\right\|_{L^\infty{\scriptscriptstyle\left(\mathbb{R}^3\smallsetminus B\left(y^k,\frac{R_k}{4}\right)\right)}}+\left\|\tau^k\right\|_{L^2(\mathbb{R}^3)}\nonumber\\
&=\left\|\tau^k\right\|_{L^2(\mathbb{R}^3)}
\end{align}
since $\mathrm{supp}\ \xi_{\frac{R_k}{2},y^k}=B\left(y^k,\frac{R_k}{4}\right)$. Then, to prove that $(H_{p,\Psi_1^k}-i\eta)^{-1}\psi^k_{i,1}$ converges to $0$ in $L^2\left(\mathbb{R}^3\smallsetminus B\left(y^k,\frac{R_k}{4}\right)\right)$, it is enough to show that the norm of the commutator $$\left[(H_{p,\Psi_1^k}-i\eta)^{-1},\xi_{\frac{R_k}{8},y^k}\right]$$ converges to $0$. We remark that 
\begin{align*}
\tau^k&=(H_{p,\Psi_1^k}-i\eta)^{-1}\left[\xi_{\frac{R_k}{8},y^k},(H_{p,\Psi_1^k}-i\eta)\right](H_{p,\Psi_1^k}-i\eta)^{-1}\psi^k_{i}\\
&=(H_{p,\Psi_1^k}-i\eta)^{-1}i\bm{\alpha}\cdot\nabla\xi_{\frac{R_k}{8},y^k}(H_{p,\Psi_1^k}-i\eta)^{-1}\psi^k_{i}.
\end{align*}
Hence, using lemma \ref{lemdomination}, we obtain
\begin{equation}\label{eqnormtaupsi1}
\|\tau^k\|_{H^{1/2}(\mathbb{R}^3)}=O(R_k^{-1}).
\end{equation}
Finally, using  the fact that $R_k\xrightarrow[k]{}+\infty$, 
\begin{equation}\label{eqconvnormpsi1ball}
\left\|(H_{p,\Psi_1^k}-i\eta)^{-1}\psi^k_{i,1}\right\|_{L^2\left(\mathbb{R}^3\smallsetminus B\left(y^k,\frac{R_k}{4}\right)\right)}\xrightarrow[k]{}0 
\end{equation}
and
\begin{equation}\label{eqconvnormpsi1ballp}
\left\|(H_{p,\Psi_1^k}-i\eta)^{-1}\psi^k_{i,1}\right\|_{L^p\left(\mathbb{R}^3\smallsetminus B\left(y^k,\frac{R_k}{4}\right)\right)}\xrightarrow[k]{}0 
\end{equation}
for $2\le p<3$ thanks to interpolation inequality and Sobolev embeddings. Indeed, we remind that $(H_{p,\Psi_1^k}-i\eta)^{-1}\psi^k_{i,1}$ is bounded in $H^{1/2}(\mathbb{R}^3)$.

Second, we consider the potential 
\begin{eqnarray}\label{eqpotenital}
W^k:=H_{p,\Psi_1^k}-H_{p,\Psi^k}&=&\left[-\beta \frac{g_\sigma^2}{4\pi}\sum\limits_{j=1}^A\left(\frac{e^{-m_\sigma|\cdot|}}{|\cdot|}\star\left(\bar{\psi}^k_{j,1}\psi^k_{j,1}-\bar{\psi}^k_j\psi^k_j\right)\right)\right.\nonumber\\
&&\left.+\frac{g_\omega^2}{4\pi}\sum\limits_{j=1}^A\left(\frac{e^{-m_\omega|\cdot|}}{|\cdot|}\star\left(\left|{\psi}^k_{j,1}\right|^2-\left|{\psi}^k_{j}\right|^2\right)\right)\right.\nonumber\\
&&\left.+\frac{g_\rho^2}{4\pi}\sum\limits_{j=1}^Z\left(\frac{e^{-m_\rho|\cdot|}}{|\cdot|}\star\left(\left|{\psi}^k_{j,1}\right|^2-\left|{\psi}^k_{j}\right|^2\right)\right)\right.\nonumber\\
&&\left.-\frac{g_\rho^2}{4\pi}\sum\limits_{j=Z+1}^A\left(\frac{e^{-m_\rho|\cdot|}}{|\cdot|}\star\left(\left|{\psi}^k_{j,1}\right|^2-\left|{\psi}^k_{j}\right|^2\right)\right)\right.\nonumber\\
&&\left.+\frac{e^2}{4\pi}\sum\limits_{j=1}^Z\left(\frac{1}{|\cdot|}\star\left(\left|{\psi}^k_{j,1}\right|^2-\left|{\psi}^k_{j}\right|^2\right)\right)\right]
\end{eqnarray}
and we estimate the $L^{7}$-norm of $W^k$ in $B_k:=B(y^k,\frac{R_k}{4})$. Using (\ref{eqdichotomydefconv}) and the definitions of ${\psi}^k_{j,1}$ and ${\psi}^k_{j,2}$, we obtain, for $1\le p <\frac{3}{2}$,
\begin{align*}
\left\|\bar{\psi}^k_{j,1}\psi^k_{j,1}-\bar{\psi}^k_j\psi^k_j\right\|_{L^p(B_k)}\le&\left\|\bar{\psi}^k_{j,1}\psi^k_{j,1}-(\bar{\psi}^k_{j,1}+\bar\psi^k_{j,2})(\psi^k_{j,1}+\psi^k_{j,2})\right\|_{L^{p}(B_k)}\\
&+C \left\|\psi^k_{j,1}+\psi^k_{j,2}-\psi^k_{j}\right\|_{L^{2p}\left( \mathbb{R}^3\right)}\xrightarrow[k]{} 0
\end{align*}
and, in the same way,
\begin{align*}
\left\|\left|{\psi}^k_{j,1}\right|^2-\left|{\psi}^k_{j}\right|^2\right\|_{L^p(B_k)}\le&\left\|\left|\psi^k_{j,1}\right|^2-\left|(\psi^k_{j,1}+\psi^k_{j,2})\right|^2\right\|_{L^{p}(B_k)}\\
&+C \left\|\psi^k_{j,1}+\psi^k_{j,2}-\psi^k_{j}\right\|_{L^{2p}\left( \mathbb{R}^3\right)}\xrightarrow[k]{} 0.
\end{align*}
Next, we remark that this potential contains three types of terms; for the first one, we have
\begin{align*}
&\left\|\frac{e^{-m_\sigma|\cdot|}}{|\cdot|}\star\left(\bar{\psi}^k_{j,1}\psi^k_{j,1}-\bar{\psi}^k_j\psi^k_j\right)\right\|_{L^{7}(B_k)}\\
&\le\left\|\frac{e^{-m_\sigma|\cdot|}}{|\cdot|}\right\|_{L^{{35/12}}(B_k)}\left\|\left(\bar{\psi}^k_{j,1}\psi^k_{j,1}-\bar{\psi}^k_j\psi^k_j\right)\right\|_{L^{5/4}(B_k)}\xrightarrow[k]{}0.
\end{align*}
Similarly, for the second type of terms, we obtain
\begin{align*}
&\left\|\frac{e^{-m_\omega|\cdot|}}{|\cdot|}\star\left(\left|{\psi}^k_{j,1}\right|^2-\left|{\psi}^k_{j}\right|^2\right)\right\|_{L^{7}(B_k)}\\
&\le\left\|\frac{e^{-m_\omega|\cdot|}}{|\cdot|}\right\|_{L^{{35/12}}(B_k)}\left\|\left(\left|{\psi}^k_{j,1}\right|^2-\left|{\psi}^k_{j}\right|^2\right)\right\|_{L^{5/4}(B_k)}\xrightarrow[k]{}0.
\end{align*}
For the last term, we remind that $\frac{1}{|x|}$ can be written as $\frac{1}{|x|}=h_1(x)+h_2(x)$ with $h_1\in L^{{35/12}}(\mathbb{R}^3)$ and $h_2\in L^{7}(\mathbb{R}^3)$, where $h_1(x)=\frac{1}{|x|}$ for $|x|\le 1$, $h_1(x)=0$ otherwise. Then
\begin{align*}
&\left\|\frac{1}{|\cdot|}\star\left(\left|{\psi}^k_{j,1}\right|^2-\left|{\psi}^k_{j}\right|^2\right)\right\|_{L^{7}(B_k)}\\
&\le \left\|h_1\right\|_{L^{{35/12}}(B_k)}\left\|\left(\left|{\psi}^k_{j,1}\right|^2-\left|{\psi}^k_{j}\right|^2\right)\right\|_{L^{5/4}(B_k)}\\
&+\left\|h_2\right\|_{L^{7}(B_k)}\left\|\left(\left|{\psi}^k_{j,1}\right|^2-\left|{\psi}^k_{j}\right|^2\right)\right\|_{L^1(B_k)}\xrightarrow[k]{}0.
\end{align*}
Finally,
\begin{equation}\label{eqnormpot}
\|W^k\|_{L^{7}\left(B\left(y^k,\frac{R_k}{4}\right)\right)}\xrightarrow[k]{}0.
\end{equation}
In conclusion, using (\ref{eqconvnormpsi1ballp}) and (\ref{eqnormpot}), we obtain
\begin{align*}
&\left\|(H_{p,\Psi_1^k}-H_{p,\Psi^k})(H_{p,\Psi_1^k}-i\eta)^{-1}\psi^k_{i,1}\right\|_{L^2(\mathbb{R}^3)}^2=\left\|W^k(H_{p,\Psi_1^k}-i\eta)^{-1}\psi^k_{i,1}\right\|^2_{L^2(\mathbb{R}^3)}\\
&\le\left\|W^k(H_{p,\Psi_1^k}-i\eta)^{-1}\psi^k_{i,1}\right\|_{L^2\left( B\left(y^k,\frac{R_k}{4}\right)\right)}^2\\
&+\left\|W^k(H_{p,\Psi_1^k}-i\eta)^{-1}\psi^k_{i,1}\right\|_{L^2\left(\mathbb{R}^3\smallsetminus B\left(y^k,\frac{R_k}{4}\right)\right)}^2\\
&\le\left\|W^k\right\|_{L^{7}\left( B\left(y^k,\frac{R_k}{4}\right)\right)}^2\left\|(H_{p,\Psi_1^k}-i\eta)^{-1}\psi^k_{i,1}\right\|_{L^{14/5}(\mathbb{R}^3)}^2\\
&+\left\|W^k\right\|_{L^{7}(\mathbb{R}^3)}^2\left\|(H_{p,\Psi_1^k}-i\eta)^{-1}\psi^k_{i,1}\right\|_{L^{14/5}\left(\mathbb{R}^3\smallsetminus B\left(y^k,\frac{R_k}{4}\right)\right)}^2\xrightarrow[k]{}0.
\end{align*}
\end{proof}
Hence, if we apply lemma \ref{lemdomination} to $\varphi=(H_{p,\Psi_1^k}-H_{p,\Psi^k})(H_{p,\Psi_1^k}-i\eta)^{-1}\psi^k_{i,1}$ and we use the result of lemma \ref{lemproj}, we can conclude that
\begin{equation}\label{eqconvH12}
\left\|(H_{p,\Psi^k}-i\eta)^{-1}(H_{p,\Psi_1^k}-H_{p,\Psi^k})(H_{p,\Psi_1^k}-i\eta)^{-1}\psi^k_{i,1}\right\|_{H^{1/2}}\xrightarrow[k]{}0
\end{equation}
for all $\eta \in \mathbb{R}$.

Finally, to prove that $\int_{-\infty}^{+\infty}f^k(\eta)\,d\eta\to 0$ as $k\to \infty$, we use the Lebesgue's dominated convergence theorem.
Indeed, the sequence  $f^k$ converges to $f=0$ for all $\eta\in \mathbb{R}$ and is dominated by an integrable function $g$. In particular, using lemma \ref{lemdomination} and its proof, we remark that, $\forall k\in \mathbb{N}$,
\begin{align*}
\left|f^k(\eta)\right|&\le\frac{1}{(m_b^2+\eta^2)^{1/4}}\left(1+\frac{C}{(\hat h_p^2+\eta^2)^{1/3}}\right)\left\|W^k(H_{p,\Psi_1^k}-i\eta)^{-1}\psi^k_{i,1}\right\|_{L^2}\\
&\le\frac{\tilde C}{(m_b^2+\eta^2)^{1/4}(\hat h_p^2+\eta^2)^{1/3}}\left(1+\frac{C}{(\hat h_p^2+\eta^2)^{1/3}}\right)\left\|\psi^k_{i,1}\right\|_{L^2}:=g(\eta)\\
&\le\frac{\tilde C}{(m_b^2+\eta^2)^{1/4}(\hat h_p^2+\eta^2)^{1/3}}\left(1+\frac{C}{(\hat h_p^2+\eta^2)^{1/3}}\right):=g(\eta)
\end{align*}
and $g(\eta)\in L^1(\mathbb{R})$. Then
$$\int_{-\infty}^{+\infty}f^k(\eta)\,d\eta\xrightarrow[k]{}0.$$
Now, to prove that $\Lambda^-_{p,\Psi^k}\psi^k_{i,1}=\Lambda^-_{p,\Psi^k}\xi_{\frac{R_k}{8}}(\cdot-y^k)\psi^k_{i}$ converges to $0$ in $H^{1/2}(\mathbb{R}^3)$, we give an estimate on the commutator 
$\left[\Lambda^-_{p,\Psi^k},\xi_{\frac{R_k}{8}}(\cdot-y^k)\right]$. Writing $\xi_{\frac{R_k}{8},y^k}(\cdot)=\xi_{\frac{R_k}{8}}(\cdot-y^k)$ and using Cauchy’s formula, we infer
\begin{align*}
\left[\Lambda^-_{p,\Psi^k},\xi_{\frac{R_k}{8},y^k}\right]&=-\frac{1}{2\pi}\int_{-\infty}^{+\infty}\frac{1}{H_{p,\Psi^k}+i\eta}\xi_{\frac{R_k}{8},y^k}-\xi_{\frac{R_k}{8},y^k}\frac{1}{H_{p,\Psi^k}+i\eta}\,d\eta\\
&=-\frac{1}{2\pi}\int_{-\infty}^{+\infty}\frac{1}{H_{p,\Psi^k}+i\eta}\left[\xi_{\frac{R_k}{8},y^k},H_{p,\Psi^k}+i\eta\right]\frac{1}{H_{p,\Psi^k}+i\eta}\,d\eta\\
&=-\frac{i}{2\pi}\int_{-\infty}^{+\infty}\frac{1}{H_{p,\Psi^k}+i\eta}\bm{\alpha}\cdot\nabla\xi_{\frac{R_k}{8},y^k}\frac{1}{H_{p,\Psi^k}+i\eta}\,d\eta.
\end{align*}
Hence,
$$
\left\|\left[\Lambda^-_{p,\Psi^k},\xi_{\frac{R_k}{8},y^k}\right]\psi^k_{i}\right\|_{H^{1/2}}\le C\int_{-\infty}^{+\infty}\frac{\left\|\nabla\xi_{\frac{R_k}{8},y^k}\right\|_{L^{\infty}}}{(m_b^2+\eta^2)^{1/4}(m_b^2+\eta^2)^{1/2}}\,d\eta=O(R_k^{-1}).
$$
Then, as $R_k\xrightarrow[k]{}+\infty$, we obtain  
$$\left\|\left[\Lambda^-_{p,\Psi^k},\xi_{\frac{R_k}{8},y^k}\right]\psi^k_{i}\right\|_{H^{1/2}}\xrightarrow[k]{}0$$
and, since $\Lambda^-_{p,\Psi^k}\psi^k_{i,1}=\left[\Lambda^-_{p,\Psi^k},\xi_{\frac{R_k}{8},y^k}\right]\psi^k_{i}+\xi_{\frac{R_k}{8},y^k}\Lambda^-_{p,\Psi^k}\psi^k_{i}$ and $\Lambda^-_{p,\Psi^k}\psi^k_{i}=0$, we conclude that
$$ \Lambda^-_{p,\Psi_1^k}(\psi^k_{1,1},\ldots,\psi^k_{Z,1})\xrightarrow[k]{}0$$
 in $\left(H^{1/2}(\mathbb{R}^3)\right)^Z$. 
Moreover, with the same arguments used above, we prove that
$$
 \Lambda^-_{n,\Psi_1^k}(\psi^k_{Z+1,1},\ldots,\psi^k_{A,1})\xrightarrow[k]{}0
$$
in $\left(H^{1/2}(\mathbb{R}^3)\right)^N$.

Furthermore, to show that
$$
\begin{array}{lll}
\Lambda^-_{p,\Psi_2^k}(\psi^k_{1,2},\ldots,\psi^k_{Z,2})\xrightarrow[k]{}0
 &\mbox{and} &\Lambda^-_{n,\Psi_2^k}(\psi^k_{Z+1,2},\ldots,\psi^k_{A,2})\xrightarrow[k]{}0
\end{array}
$$
in $\left(H^{1/2}(\mathbb{R}^3)\right)^Z$ and $\left(H^{1/2}(\mathbb{R}^3)\right)^N$ respectively, we can proceed as before; only the proof of 
\begin{equation}
\begin{array}{ll}
(H_{p,\Psi_2^k}-H_{p,\Psi^k})(H_{p,\Psi_2^k}-i\eta)^{-1}\psi^k_{i,1}\xrightarrow[k]{L^2}0, & \forall \eta \in \mathbb{R},
\end{array}
\end{equation}
is slightly different. In this case, $$\|H_{p,\Psi_2^k}-H_{p,\Psi^k}\|_{L^{7}\left(\mathbb{R}^3\smallsetminus B\left(y^k,\frac{R_k}{4}\right)\right)}\xrightarrow[k]{}0,$$ thanks to the localization property of $\psi^k_{i,1}$, and $(H_{p,\Psi_2^k}-i\eta)^{-1}\psi^k_{i,2}$ converges strongly to zero in $L^p\left(B\left(y^k,\frac{R_k}{4}\right)\right)$ for $2\le p<3$. In conclusion,
\begin{align*}
&\left\|(H_{p,\Psi_2^k}-H_{p,\Psi^k})(H_{p,\Psi_2^k}-i\eta)^{-1}\psi^k_{i,2}\right\|_{L^2(\mathbb{R}^3)}^2 \\
&\le 
 C_1\left\|(H_{p,\Psi_2^k}-i\eta)^{-1}\psi^k_{i,2}\right\|_{L^{14/5}\left( B\left(y^k,\frac{R_k}{4}\right)\right)}^2\\
&+C_2\left\|(H_{p,\Psi_2^k}-H_{p,\Psi^k})\right\|_{L^{7}\left(\mathbb{R}^3\smallsetminus B\left(y^k,\frac{R_k}{4}\right)\right)}^2\xrightarrow[k]{}0.
\end{align*}

Now, we want to construct $\Phi_1^k=(\Phi_{p,1}^k,\Phi_{n,1}^k), \Phi_2^k=(\Phi_{p,2}^k,\Phi_{n,2}^k)\in \left(H^{1/2}(\mathbb R^3)\right)^Z\times\left(H^{1/2}(\mathbb R^3)\right)^N$, small perturbations of $\Psi_1^k,\Psi_2^k$ in $\left(H^{1/2}(\mathbb R^3)\right)^A$, that satisfy the constraints of $I\left(\lambda_1,\ldots,\lambda_A\right)$ and $I\left(1-\lambda_1,\ldots,1-\lambda_A\right)$ respectively. For this purpose, we use the following lemma and its corollary. The proofs of the lemma and the corollary are given in the appendix.
\begin{lem}\label{lemimplicitfunction} Take $\Psi=(\Psi_{p},\Psi_{n})\in \left(H^{1/2}(\mathbb R^3)\right)^Z \times \left(H^{1/2}(\mathbb R^3)\right)^N$ such that 
\begin{enumerate}
\item \label{h1lemimplicitfunction}$\mathrm{Gram}_{L^2}\left(\Psi_p\right):=G_p\leq\mathbbm{1}_Z$ and $\mathrm{Gram}_{L^2}\left(\Psi_n\right):=G_n\leq\mathbbm{1}_N$ are invertible matrices;
\item \label{h2lemimplicitfunction}$$\begin{array}{l}\left\|\Lambda^-_{p,\Psi}\Psi_p\right\|_{(H^{1/2})^Z}\leq \tilde{\delta}\\[6pt]\left\|\Lambda^-_{n,\Psi}\Psi_n\right\|_{(H^{1/2})^N}\leq \tilde{\delta}\end{array}$$
for $\tilde{\delta}>0$ small enough.
\end{enumerate}
 If $g_\sigma,g_\omega,g_\rho$ and $e$ are sufficiently small, there exists $\Phi=(\Phi_{p},\Phi_{n})\in \left(H^{1/2}(\mathbb R^3)\right)^Z\times  \left(H^{1/2}(\mathbb R^3)\right)^N$ such that
 \begin{eqnarray}\label{eqconstraintproton}
\Lambda^-_{p,\Phi}\Phi_{p}=0,\\\label{eqconstraintneutron}
\Lambda^-_{n,\Phi}\Phi_{n}=0.
\end{eqnarray}
Moreover,
\begin{eqnarray}\label{eqgramproton}
\mathrm{Gram}_{L^2}\left(\Phi_p\right)=G_p,\\ \label{eqgramneutron}
\mathrm{Gram}_{L^2}\left(\Phi_n\right)=G_n.
\end{eqnarray}
\end{lem}
\begin{cor}\label{corimplicitfunction}
Take $\Psi^k=(\Psi_{p}^k,\Psi_{n}^k)\in \left(H^{1/2}(\mathbb R^3)\right)^Z \times \left(H^{1/2}(\mathbb R^3)\right)^N$ a sequence of functions bounded in $ \left(H^{1/2}(\mathbb R^3)\right)^A$ such that 
\begin{enumerate}
\item $\mathrm{Gram}_{L^2}\left(\Psi_p^k\right):=G_p\leq\mathbbm{1}_Z$ and $\mathrm{Gram}_{L^2}\left(\Psi_n^k\right):=G_n\leq\mathbbm{1}_N$ are invertible matrices that do not depend on $k$ for any $k\in\mathbb{N}$;
\item $$\begin{array}{l}\left\|\Lambda^-_{p,\Psi^k}\Psi_p^k\right\|_{(H^{1/2})^Z}\xrightarrow[k]{}0\\[6pt]\left\|\Lambda^-_{n,\Psi^k}\Psi_n^k\right\|_{(H^{1/2})^N}\xrightarrow[k]{}0.\end{array}$$
\end{enumerate}
 If $g_\sigma,g_\omega,g_\rho$ and $e$ are sufficiently small,  there is a constant $k_0\in \mathbb{N}$ such that, for any $k\ge k_0$, there exists $\Phi^k=(\Phi_{p}^k,\Phi_{n}^k)\in \left(H^{1/2}(\mathbb R^3)\right)^Z\times  \left(H^{1/2}(\mathbb R^3)\right)^N$ with the following properties:
 \begin{enumeratenew}[\alph]
 \item \begin{eqnarray}\label{eqconstraintprotonk}
\Lambda^-_{p,\Phi^k}\Phi_{p}^k=0,\\\label{eqconstraintneutronk}
\Lambda^-_{n,\Phi^k}\Phi_{n}^k=0.
\end{eqnarray}
\item \begin{eqnarray}\label{eqconvprotonk}
&&\left\|\Phi_{p}^k-\Psi_{p}^k\right\|_{(H^{1/2})^Z}\xrightarrow[k]{}0,\\
\label{eqconvneutronk}
&&\left\|\Phi_{n}^k-\Psi_{n}^k\right\|_{(H^{1/2})^N}\xrightarrow[k]{}0.
\end{eqnarray}
\item \begin{eqnarray}\label{eqgramprotonk}
\mathrm{Gram}_{L^2}\left(\Phi_p^k\right)=G_p,\\ \label{eqgramneutronk}
\mathrm{Gram}_{L^2}\left(\Phi_n^k\right)=G_n.
\end{eqnarray}
\end{enumeratenew}
\end{cor}
So, using the corollary \ref{corimplicitfunction}, we can conclude that
if $g_\sigma,g_\omega,g_\rho$ and $e$ are sufficiently small, 
 there is a constant $k_0\in \mathbb{N}$ such that, for any $k\ge k_0$, there exists $\Phi_1^k=(\Phi^k_{p,1},\Phi^k_{n,1})\in \left(H^{1/2}(\mathbb{R}^3)\right)^Z\times  \left(H^{1/2}(\mathbb{R}^3)\right)^N$ 
 with the following properties:
\begin{enumerate}
\item
\begin{eqnarray}\label{eqconstraintproton1}
\Lambda^-_{p,\Phi_1^k}\Phi_{p,1}^k=0,\\\label{eqconstraintneutron1}
\Lambda^-_{n,\Phi_1^k}\Phi_{n,1}^k=0.
\end{eqnarray}
\item
\begin{eqnarray}\label{eqconvproton1}
&&\left\|\Phi_{p,1}^k-\Psi_{p,1}^k\right\|_{(H^{1/2})^Z}\xrightarrow[k]{}0,\\
\label{eqconvneutron1}
&&\left\|\Phi_{n,1}^k-\Psi_{n,1}^k\right\|_{(H^{1/2})^N}\xrightarrow[k]{}0.
\end{eqnarray}
\item For $1\le i,j\le Z$, $Z+1\le i,j \le A$,
\begin{equation}
\int_{\mathbb{R}^3}\phi^{k^\ast}_{i,1}\phi^k_{j,1}=\lambda_i\delta_{ij}. \\
\end{equation}
\end{enumerate}
In particular, if $$\mathrm{Gram}_{L^2}\left(\Psi_{p,1}^k\right)=\mathrm{diag}(\lambda_1,\ldots,\lambda_Z)\ \mbox{and} \ \mathrm{Gram}_{L^2}\left(\Psi_{n,1}^k\right)=\mathrm{diag}(\lambda_{Z+1},\ldots,\lambda_A)$$ are invertible matrices, we apply the corollary \ref{corimplicitfunction} to $\Psi_1^k$.\\
On the other hand, if $\mathrm{Gram}_{L^2}\left(\Psi_{p,1}^k\right)$ or $\mathrm{Gram}_{L^2}\left(\Psi_{n,1}^k\right)$ is not an invertible matrix; then there exists $i\in\{1,\ldots,A\}$ such that $\lambda_i=0$. As a consequence, $\psi_{i,1}^k=0$ for any $k\in\mathbb{N}$.\\
We assume, without loss of generality, that $\lambda_i=0$ for $1\le i < r_p$, $Z+1\le i < r_n$ and $\lambda_i\neq 0$ for $r_p\le i\le Z$, $r_n\le i \le A$, and we denote $\hat\Psi_{p,1}^k=(\psi^k_{r_p,1},\ldots,\psi^k_{Z,1})$ and $\hat\Psi_{n,1}^k=(\psi^k_{r_n,1},\ldots,\psi^k_{A,1})$. Since $$\mathrm{Gram}_{L^2}\left(\hat\Psi_{p,1}^k\right)=\mathrm{diag}(\lambda_{r_p},\ldots,\lambda_Z)\ \mbox{and}\ \mathrm{Gram}_{L^2}\left(\hat\Psi_{n,1}^k\right)=\mathrm{diag}(\lambda_{r_n},\ldots,\lambda_A)$$ are invertible matrices, we can apply the corollary \ref{corimplicitfunction} to $\hat\Psi_1^k=(\hat\Psi_{p,1}^k,\hat\Psi_{n,1}^k)$ to obtain, for any $k\ge k_0$, $\hat\Phi_1^k=(\hat\Phi^k_{p,1},\hat\Phi^k_{n,1})$ such that 
\begin{enumerate}
\item
\begin{eqnarray*}
\Lambda^-_{p,\hat\Phi_1^k}\hat\Phi_{p,1}^k=0,\\
\Lambda^-_{n,\hat\Phi_1^k}\hat\Phi_{n,1}^k=0,
\end{eqnarray*}
\item 
\begin{eqnarray*}
&&\left\|\hat\Phi_{p,1}^k-\hat\Psi_{p,1}^k\right\|_{(H^{1/2})^{Z-r_p+1}}\xrightarrow[k]{}0,\\
&&\left\|\hat\Phi_{n,1}^k-\hat\Psi_{n,1}^k\right\|_{(H^{1/2})^{N-r_n+1}}\xrightarrow[k]{}0,
\end{eqnarray*}
\item  $\mathrm{Gram}_{L^2}\left(\hat\Psi_{p,1}^k\right)=\mathrm{Gram}_{L^2}\left(\hat\Phi_{p,1}^k\right)$ and  $\mathrm{Gram}_{L^2}\left(\hat\Psi_{n,1}^k\right)=\mathrm{Gram}_{L^2}\left(\hat\Phi_{n,1}^k\right)$.
\end{enumerate}
To conclude, it is enough to take $$\Phi_{p,1}^k=(0,\ldots,0,\hat\phi_{r_p,1}^k,\ldots,\hat\phi_{Z,1}^k)$$ and $$\Phi_{n,1}^k=(0,\ldots,0,\hat\phi_{r_n,1}^k,\ldots,\hat\phi_{A,1}^k)$$ and remark that $H_{\mu,\Phi_1^k}=H_{\mu,\hat\Phi_1^k}$ for $\mu=p,n$.

In the same way, if $g_\sigma,g_\omega,g_\rho$ and $e$ are sufficiently small, 
 there is a constant $k_0\in \mathbb{N}$ such that, for any $k\ge k_0$, there exists $\Phi_2^k=(\Phi^k_{p,2},\Phi^k_{n,2})\in \left(H^{1/2}(\mathbb{R}^3)\right)^Z\times  \left(H^{1/2}(\mathbb{R}^3)\right)^N$ 
 with the following properties:
 \begin{enumerate}
 \item
 \begin{eqnarray}\label{eqconstraintproton2}
\Lambda^-_{p,\Phi_2^k}\Phi_{p,2}^k=0,\\\label{eqconstraintneutron2}
\Lambda^-_{n,\Phi_2^k}\Phi_{n,2}^k=0.
\end{eqnarray}
\item
\begin{eqnarray}\label{eqconvproton2}
&&\left\|\Phi_{p,2}^k-\Psi_{p,2}^k\right\|_{(H^{1/2})^Z}\xrightarrow[k]{}0,\\
\label{eqconvneutron2}
&&\left\|\Phi_{n,2}^k-\Psi_{n,2}^k\right\|_{(H^{1/2})^N}\xrightarrow[k]{}0.
\end{eqnarray}
\item For $1\le i,j\le Z$, $Z+1\le i,j \le A$,
\begin{equation}
\int_{\mathbb{R}^3}\phi^{k^\ast}_{i,2}\phi^k_{j,2}=(1-\lambda_i)\delta_{ij}. \\
\end{equation}
\end{enumerate}

Using (\ref{eqconvproton1}), (\ref{eqconvneutron1}), (\ref{eqconvproton2}), (\ref{eqconvneutron2})  and the continuity of $\mathcal{E}$, we remark that
\begin{eqnarray*}
\varliminf_{k\to\infty}\mathcal{E}(\Psi_1^k)=\varliminf_{k\to\infty}\mathcal{E}(\Phi_1^k),\\
\varliminf_{k\to\infty}\mathcal{E}(\Psi_2^k)=\varliminf_{k\to\infty}\mathcal{E}(\Phi_2^k),
\end{eqnarray*}
and then, if dichotomy occurs, we have
\begin{eqnarray}\label{eqcontradictiondicho}
I&=&\lim_{k\to\infty}\mathcal{E}(\Psi^k)\ge \varliminf_{k\to\infty}\mathcal{E}(\Psi_1^k)+\varliminf_{k\to\infty}\mathcal{E}(\Psi_2^k)\nonumber\\
&=& \varliminf_{k\to\infty}\mathcal{E}(\Phi_1^k)+\varliminf_{k\to\infty}\mathcal{E}(\Phi_2^k)\nonumber\\
&\ge&I\left(\lambda_1,\ldots,\lambda_A\right)+I\left(1-\lambda_1,\ldots,1-\lambda_A\right).
\end{eqnarray}
It is now clear that (\ref{eqcontradictiondicho}) contradicts (\ref{eqinequality}).
\subsection{Vanishing does not occur}
If vanishing occurs (case ii.), then $\forall R<\infty$
$$
\sup\limits_{y\in \mathbb{R}^3}\int_{B\left(y,R\right)}\left|\psi_j^k\right|^2 \xrightarrow[k]{} 0
$$
for $j=1,\ldots,A$ and $\psi_1^k,\ldots,\psi_A^k$ converge strongly in $L^p(\mathbb{R}^3)$ to $0$ for $2<p<3$ (see lemma 7.2 of \cite{lenzmannlewin}). As a consequence,
\begin{equation*}
\lim_{k\to\infty}\mathcal{E}(\Psi^k)=\sum_{j=1}^A\lim_{k\to\infty}\int_{\mathbb{R}^3}\psi_j^{k^\ast} H_0\psi_j^k,
\end{equation*}
and
\begin{equation*}
I\left(\lambda_1,\ldots,\lambda_A\right)=m_b\sum_{j=1}^A\lambda_j
\end{equation*}
thanks to the constraints of the problem.\\
This contradicts (\ref{eqinequality}) because we have
\begin{equation*}
I=m_b A=m_b\sum_{j=1}^A\lambda_j+m_b\sum_{j=1}^A(1-\lambda_j)=I\left(\lambda_1,\ldots,\lambda_A\right)+I\left(1-\lambda_1,\ldots,1-\lambda_A\right).
\end{equation*}

At this point, we have shown that any minimizing sequence satisfies the following compactness criterion:  $\exists y^k\in \mathbb{R}^3$, $\forall \varepsilon>0$, $\exists R<\infty$
$$
\frac{1}{A}\sum_{j=1}^{A}\int_{B\left(y^k,R\right)}\left|\psi_j^k\right|^2\geq1-\varepsilon.
$$
We denote $\tilde\Psi^k=\Psi^k(\cdot+y^k)$ and we remark that the energy functional $\mathcal{E}$ is invariant by translations and $\tilde\Psi^k$ is in the minimizing set; then $\tilde\Psi^k$ is a minimizing sequence of (\ref{eqminproblem}).
Since $\tilde\Psi^k$ is bounded in $\left(H^{1/2}(\mathbb R^3)\right)^A$, $\tilde\Psi^k$ converges weakly in $\left(H^{1/2}(\mathbb R^3)\right)^A$, almost everywhere on $\mathbb R^3$ and in $\left(L^p_{loc}(\mathbb R^3)\right)^A$ for $2\le p <3$ to some $\tilde \Psi$; moreover, thanks to the concentration-compactness argument, $\tilde\Psi^k$ converges strongly to $\tilde\Psi$ in  $\left(L^{2}(\mathbb R^3)\right)^A$ and in $\left(L^p(\mathbb R^3)\right)^A$ for $2\le p <3$.\\
As $\|\tilde\psi_j-\tilde\psi_j^k\|_{L^2}\rightarrow 0$ for $k\to +\infty$, it is clear that
$$
\int_{\mathbb R^3}\tilde\psi_i^*\tilde\psi_j=\lim_{k\to +\infty}\int_{\mathbb R^3}\tilde\psi_i^{k^*}\tilde\psi_j^k=\delta_{ij}
$$
for $1\le i,j\le Z$ and $Z+1\le i,j\le A$. Furthermore, $\Lambda^-_{\mu,\tilde\Psi}\tilde\Psi_\mu=0$ for $\mu=p,n$. Indeed, as before, 
\begin{equation*}
\Lambda^-_{\mu,\tilde\Psi}\tilde\psi_{j}-\Lambda^-_{\mu,\tilde\Psi^k}\tilde\psi_{j}=\frac{1}{2\pi}\int_{-\infty}^{+\infty}({H}_{\mu,\tilde\Psi^k}-i\eta)^{-1}({V}_{\mu,\tilde\Psi}-{V}_{\mu,\tilde\Psi^k})({H}_{\mu,\tilde\Psi}-i\eta)^{-1}\tilde\psi_{j}\,d\eta
\end{equation*}
and 
$$
\left\|({H}_{\mu,\tilde\Psi^k}-i\eta)^{-1}({V}_{\mu,\tilde\Psi}-{V}_{\mu,\tilde\Psi^k})({H}_{\mu,\tilde\Psi}-i\eta)^{-1}\tilde\psi_{j}\right\|_{H^{1/2}}\xrightarrow[k]{}0
$$
since $\|\tilde\psi_j-\tilde\psi_j^k\|_{L^p}\xrightarrow[k]{} 0$ for $2\le p < 3$. Then, applying the Lebesgue's dominated convergence theorem as above, we obtain
$$
\left\|\Lambda^-_{\mu,\tilde\Psi}\tilde\psi_{j}-\Lambda^-_{\mu,\tilde\Psi^k}\tilde\psi_{j}\right\|_{H^{1/2}}\xrightarrow[k]{}0
$$
for $\mu=p$ if $1\le j\le Z$ and $\mu=n$ if $Z+1\le j\le A$. As a consequence,
$$
\left\|\Lambda^-_{\mu,\tilde\Psi}\tilde\psi_{j}-\Lambda^-_{\mu,\tilde\Psi^k}\tilde\psi_{j}\right\|_{L^{2}}\xrightarrow[k]{}0
$$
and
$$
\left\|\Lambda^-_{\mu,\tilde\Psi}\tilde\psi_{j}\right\|_{L^{2}}=\lim_{k\to+\infty}\left\|\Lambda^-_{\mu,\tilde\Psi^k}\tilde\psi_{j}\right\|_{L^{2}}=\lim_{k\to+\infty}\left\|\Lambda^-_{\mu,\tilde\Psi^k}\tilde\psi_{j}^k\right\|_{L^{2}}=0,
$$
thanks to the properties of the spectral projection $\Lambda^-_{\mu,\tilde\Psi^k}$ and using the fact that $\|\tilde\psi_j-\tilde\psi_j^k\|_{L^2}\xrightarrow[k]{} 0$.
So we can conclude that $\tilde\Psi$ satisfies the constraints of the minimization problem (\ref{eqminproblem}). 

Finally, we have to prove that
$$
\mathcal{E}(\tilde\Psi)\le\liminf_{k\to+\infty}\mathcal{E}(\tilde\Psi^k).
$$
It is clear that if $\|\tilde\psi_j-\tilde\psi_j^k\|_{L^p}\xrightarrow[k]{} 0$ for $2\le p < 3$, then
\begin{equation}\label{eqconvpot}
\left(\tilde\psi_j,V_{\mu,\tilde\Psi}\tilde\psi_j\right)=\lim_{k\to+\infty}\left(\tilde\psi_j^k,V_{\mu,\tilde\Psi^k}\tilde\psi_j^k\right)
\end{equation}
for $\mu=p$ if $1\le j\le Z$ and $\mu=n$ if $Z+1\le j\le A$. Moreover, we observe that
$$
\left\|\Lambda^-\tilde\psi_{j}-\Lambda^-\tilde\psi^k_{j}\right\|_{H^{1/2}}\le\left\|(\Lambda^--\Lambda^-_{\mu,\tilde\Psi^k})(\tilde\psi_{j}-\tilde\psi^k_{j})\right\|_{H^{1/2}}+\left\|(\Lambda^-_{\mu,\tilde\Psi^k}-\Lambda^-_{\mu,\tilde\Psi})\tilde\psi_{j}\right\|_{H^{1/2}}
$$
and, with the same arguments used above, we obtain
\begin{equation}\label{eqconvnegpart}
\left\|\Lambda^-\tilde\psi_{j}\right\|_{H^{1/2}}=\lim_{k\to+\infty}\left\|\Lambda^-\tilde\psi^k_{j}\right\|_{H^{1/2}}.
\end{equation}
Then, using (\ref{eqconvpot}), (\ref{eqconvnegpart}) and the weak lower semicontinuity of the $H^{1/2}$-norm, we get
\begin{eqnarray*}
\mathcal{E}(\tilde\Psi)&=&\sum_{j=1}^A\left(\tilde\psi_j,| H_0|\tilde\psi_j\right)_{L^2}-2\sum_{j=1}^A\left(\Lambda^-\tilde\psi_j,| H_0|\Lambda^-\tilde\psi_j\right)_{L^2}\\
&&+\frac{1}{2}\sum_{j=1}^Z\left(\tilde\psi_j, V_{p,\tilde\Psi}\tilde\psi_j\right)_{L^2}+\frac{1}{2}\sum_{j=Z+1}^A\left(\tilde\psi_j, V_{n,\tilde\Psi}\tilde\psi_j\right)_{L^2}\\
&\le&\liminf_{k\to+\infty}\mathcal{E}(\tilde\Psi^k)\le \mathcal{E}(\tilde\Psi).
\end{eqnarray*}
As a conclusion, $\tilde\Psi$ is a minimizer of (\ref{eqminproblem}) and the minimizing sequence $\Psi^k$ is relatively compact in $(H^{1/2})^A$ up to a translation.

\subsection{The subadditivity condition}
To conclude the proof of theorem \ref{thinequality}, it remains to show that the strict subadditivity condition (\ref{eqinequality}) is a necessary condition for the compactness of all minimizing sequences (see \cite{lionscc1}, \cite{lionscc2}). 

First of all, we prove that we always have
\begin{equation}\label{eqinequalitylarge}
I\le I\left(\lambda_1,\ldots,\lambda_A\right)+I\left(1-\lambda_1,\ldots,1-\lambda_A\right)
\end{equation}
for all $\lambda_k \in [0,1]$, $k=1,\ldots,A$, such that $\sum\limits_{k=1}^A\lambda_k\in (0,A)$.

Let $\varepsilon>0$ and $\Psi_1^\varepsilon$, $\Psi_2^\varepsilon$ be satisfy
\begin{equation}
\left\{
\begin{aligned}
&I\left(\lambda_1,\ldots,\lambda_A\right)\le \mathcal{E}\left(\Psi_1^\varepsilon\right)\le I\left(\lambda_1,\ldots,\lambda_A\right)+\varepsilon,\\
&\mathrm{Gram}_{L^2}(\Psi_{p,1}^\varepsilon)=\mathrm{diag}\left(\lambda_1,\ldots,\lambda_Z\right),\\
&\mathrm{Gram}_{L^2}(\Psi_{n,1}^\varepsilon)=\mathrm{diag}\left(\lambda_{Z+1},\ldots,\lambda_A\right),\\
&\Lambda^-_{p,\Psi_1^\varepsilon}\Psi_{p,1}^\varepsilon=0,\ \Lambda^-_{n,\Psi_1^\varepsilon}\Psi_{n,1}^\varepsilon=0
\end{aligned}
\right.
\end{equation}
and
\begin{equation}
\left\{
\begin{aligned}
&I\left(1-\lambda_1,\ldots,1-\lambda_A\right)\le \mathcal{E}\left(\Psi_2^\varepsilon\right)\le I\left(1-\lambda_1,\ldots,1-\lambda_A\right)+\varepsilon,\\
&\mathrm{Gram}_{L^2}(\Psi_{p,2}^\varepsilon)=\mathrm{diag}\left(1-\lambda_1,\ldots,1-\lambda_Z\right),\\
& \mathrm{Gram}_{L^2}(\Psi_{n,2}^\varepsilon)=\mathrm{diag}\left(1-\lambda_{Z+1},\ldots,1-\lambda_A\right),\\
&\Lambda^-_{p,\Psi_2^\varepsilon}\Psi_{p,2}^\varepsilon=0,\ \Lambda^-_{n,\Psi_2^\varepsilon}\Psi_{n,2}^\varepsilon=0.
\end{aligned}
\right.
\end{equation}
By a density argument, we may assume that $\Psi_1^\varepsilon$ and $\Psi_2^\varepsilon$ have compact support and we denote by $\Psi_2^{\varepsilon,k}=\Psi_2^{\varepsilon}(\cdot+k\eta)$ where $\eta$ is some given unit vector in $\mathbb{R}^3$. Since for $k$ large enough the distance between the supports of $\Psi_1^\varepsilon$ and $\Psi_2^{\varepsilon,k}$ is strictly positive and goes to $+\infty$ as $k$ goes to $+\infty$, we deduce
\begin{equation}
\left\{
\begin{aligned}
&\mathcal{E}(\Psi^{\varepsilon,k})-\left[\mathcal{E}(\Psi_1^{\varepsilon})+\mathcal{E}(\Psi_2^{\varepsilon,k})\right]\xrightarrow[k]{}0,\\
&\begin{aligned}
&\int_{\mathbb{R}^3}\psi_i^{{\varepsilon,k}^*}\psi_j^{\varepsilon,k}\xrightarrow[k]{}\delta_{ij} &1\le i,j\le Z,~Z+1\le i,j\le A\\
&\left\|\Lambda^-_{p,\Psi^{\varepsilon,k}}\psi^{\varepsilon,k}_j\right\|_{H^{1/2}}\xrightarrow[k]{}0 &j=1,\ldots,Z,\\
&\left\|\Lambda^-_{n,\Psi^{\varepsilon,k}}\psi^{\varepsilon,k}_j\right\|_{H^{1/2}}\xrightarrow[k]{}0 &j=Z+1,\dots,A
\end{aligned}
\end{aligned}
\right.
\end{equation}
with $\Psi^{\varepsilon,k}=\Psi_1^\varepsilon+\Psi_2^{\varepsilon,k}$. Indeed, $\Lambda^-_{\mu,\Psi^{\varepsilon,k}}\psi^{\varepsilon,k}_j=\Lambda^-_{\mu,\Psi^{\varepsilon,k}}\psi^{\varepsilon}_{j,1}-\Lambda^-_{\mu,\Psi^{\varepsilon}_1}\psi^{\varepsilon}_{j,1}+\Lambda^-_{\mu,\Psi^{\varepsilon,k}}\psi^{\varepsilon,k}_{j,2}-\Lambda^-_{\mu,\Psi^{\varepsilon,k}_2}\psi^{\varepsilon,k}_{j,2}$ and, by arguments similar to those used above, we obtain  
\begin{align*}
&\left\|\Lambda^-_{\mu,\Psi^{\varepsilon,k}}\psi^{\varepsilon}_{j,1}-\Lambda^-_{\mu,\Psi^{\varepsilon}_1}\psi^{\varepsilon}_{j,1}\right\|_{H^{1/2}}\xrightarrow[k]{}0,\\
&\left\|\Lambda^-_{\mu,\Psi^{\varepsilon,k}}\psi^{\varepsilon,k}_{j,2}-\Lambda^-_{\mu,\Psi^{\varepsilon,k}_2}\psi^{\varepsilon,k}_{j,2}\right\|_{H^{1/2}}\xrightarrow[k]{}0
\end{align*}
for $\mu=p,n$. Then, as before, we can construct $\Phi^{\varepsilon,k}$, small perturbation of $\Psi^{\varepsilon,k}$ in $\left(H^{1/2}(\mathbb{R}^3)\right)^A$, such that
\begin{equation}
\left\{
\begin{aligned}
&\begin{aligned}
I\left(\lambda_1,\ldots,\lambda_A\right)&+I\left(1-\lambda_1,\ldots,1-\lambda_A\right)\le\lim_{k\to+\infty} \mathcal{E}\left(\Phi^{\varepsilon,k}\right)= \lim_{k\to+\infty} \mathcal{E}\left(\Psi^{\varepsilon,k}\right)\\&=\mathcal{E}\left(\Psi_1^{\varepsilon}\right)+\mathcal{E}\left(\Psi_2^{\varepsilon}\right)\le I\left(\lambda_1,\ldots,\lambda_A\right)+I\left(1-\lambda_1,\ldots,1-\lambda_A\right)+2\varepsilon,\\
\end{aligned}\\
&\mathrm{Gram}_{L^2}(\Phi_{p}^{\varepsilon,k})=\mathbbm{1}_Z,\ \mathrm{Gram}_{L^2}(\Phi_{n}^{\varepsilon,k})=\mathbbm{1}_N,\\
&\Lambda^-_{p,\Phi^{\varepsilon,k}}\Phi_{p}^{\varepsilon,k}=0,\ \Lambda^-_{n,\Phi^{\varepsilon,k}}\Phi_{n}^{\varepsilon,k}=0
\end{aligned}
\right.
\end{equation}
and, by definition of $I$, we conclude
$$
I\le I\left(\lambda_1,\ldots,\lambda_A\right)+I\left(1-\lambda_1,\ldots,1-\lambda_A\right)+2\varepsilon.
$$

In fact, this argument prove also that if $$I= I\left(\lambda_1,\ldots,\lambda_A\right)+I\left(1-\lambda_1,\ldots,1-\lambda_A\right),$$ then there exists a minimizing sequence that is not relatively compact. Indeed, let $\Psi_1^k$ and $\Psi_2^k$ be minimizing sequences of $I\left(\lambda_1,\ldots,\lambda_A\right)$ and $I\left(1-\lambda_1,\ldots,1-\lambda_A\right)$
respectively, with compact support and such that $\mathrm{dist}\left(\mathrm{supp}\ \psi_{i,1}^k, \mathrm{supp}\ \psi_{i,2}^k  \right) \xrightarrow[k]{} +\infty$ for $i=1,\ldots,A$. If we take $\psi_i^{k}=\psi_{i,1}^k+\psi_{i,2}^k$, we can show that $\int_{\mathbb{R}^3}\psi_i^{k}(x)\chi(x)\,dx\xrightarrow[k]{}0$ for all $\chi\in\mathcal{D}(\mathbb{R}^3)$; then $\Psi^k$ converges weakly to $0$ in $(H^{1/2})^A$. Now, as before, we can construct $\Phi^k$, small perturbation of $\Psi^k$ in $(H^{1/2})^A$, which is in the minimizing set of $I$ and such that
$$
 \lim_{k\to+\infty}\mathcal{E}\left(\Phi^{k}\right)=I\left(\lambda_1,\ldots,\lambda_A\right)+I\left(1-\lambda_1,\ldots,1-\lambda_A\right)=I.
$$
As a conclusion, $\Phi^{k}$ is a minimizing sequence that cannot be relatively compact.

\section{Solutions of the relativistic mean-field equations}\label{appsolutions}
In this section, we prove that, in a weakly relativistic regime, a minimizer of (\ref{eqminproblem}) is a solution of the equations (\ref{eqdiracyukawaproton}) and (\ref{eqdiracyukawaneutron}).

Let 
\begin{equation}\label{eqspaceX}
X=\left\{\gamma\in\mathcal{B}(\mathcal{H});\gamma=\gamma^*, (m_b^2-\Delta)^{1/4}\gamma(m_b^2-\Delta)^{1/4}\in \sigma_{1}(\mathcal{H}) \right\}
\end{equation}
where $\mathcal{B}(\mathcal{H})$ is the space of bounded linear maps from $\mathcal{H}$ to $\mathcal{H}$ and $\sigma_{1}(\mathcal{H})$ is the space of trace-class operators on $\mathcal{H}$.\\
Now, to each $P\in \mathbb{N}$, we associate 
\begin{equation}\label{eqespacegamma}
\Gamma_P=\left\{\gamma\in X; \gamma^2=\gamma, \tr(\gamma)=P\right\}.
\end{equation}
Given $\gamma=(\gamma_p,\gamma_n)\in X\times X$, we define
\begin{eqnarray}\label{eqprojdiracyukawaproton}
H_{ p,\gamma}\gamma_p&:=&\left[H_0-\beta \frac{g_\sigma^2}{4\pi}\left(\frac{e^{-m_\sigma|\cdot|}}{|\cdot|}\star\rho_s\right)+\frac{g_\omega^2}{4\pi}\left(\frac{e^{-m_\omega|\cdot|}}{|\cdot|}\star\rho_0\right)\right.\\
&&\left.+\frac{g_\rho^2}{4\pi}\left(\frac{e^{-m_\rho|\cdot|}}{|\cdot|}\star\rho_{00}\right)+\frac{e^2}{4\pi}\left(\frac{1}{|\cdot|}\star\rho_p\right)\right]\gamma_p\nonumber
\end{eqnarray}
\begin{eqnarray}\label{eqprojdiracyukawaneutron}
H_{ n,\gamma}\gamma_{n}&:=&\left[H_0-\beta \frac{g_\sigma^2}{4\pi}\left(\frac{e^{-m_\sigma|\cdot|}}{|\cdot|}\star\rho_s\right)+\frac{g_\omega^2}{4\pi}\left(\frac{e^{-m_\omega|\cdot|}}{|\cdot|}\star\rho_0\right)\right.\\
&&\left.-\frac{g_\rho^2}{4\pi}\left(\frac{e^{-m_\rho|\cdot|}}{|\cdot|}\star\rho_{00}\right)\right]\gamma_n\nonumber
\end{eqnarray} 
where
\begin{eqnarray*}
\rho_s(x)&=&\bar\rho_p(x)+\bar\rho_n(x)\\
\rho_0(x)&=&\rho_p(x)+\rho_n(x)\\
\rho_{00}(x)&=&\rho_p(x)-\rho_n(x)
\end{eqnarray*}
with $\bar\rho_{p}(x)=\tr(\beta\gamma_{p}(x,x))$, $\bar\rho_{n}(x)=\tr(\beta\gamma_{n}(x,x))$, $\rho_{p}(x)=\tr(\gamma_{p}(x,x))$ and $\rho_{n}(x)=\tr(\gamma_{n}(x,x))$.\\
Finally, we define
\begin{eqnarray*}
\Lambda_{p,\gamma}^{\pm}&=&\chi_{\mathbb{R}^{\pm}}(H_{p,\gamma}),\\
\Lambda_{n,\gamma}^{\pm}&=&\chi_{\mathbb{R}^{\pm}}(H_{n,\gamma}).
\end{eqnarray*}

Let $\tilde\Psi=(\tilde\Psi_p,\tilde\Psi_n)$ be a minimizer of the problem (\ref{eqminproblem}); to prove that $\tilde\psi_i$ is a solution of (\ref{eqdiracyukawaproton}) for $1\le i\le Z$ and of (\ref{eqdiracyukawaneutron}) for $Z+1\le i\le A$, we proceed as follow: first, we consider $\tilde\gamma_p$ and $\tilde\gamma_n$ the orthogonal projectors defined by 
\begin{equation}\label{eqprojproton}
\tilde\gamma_p=\sum_{i=1}^{Z}\ket{\tilde\psi_i}\bra{\tilde\psi_i}
\end{equation}
and
\begin{equation}
\label{eqprojneutron}
\tilde\gamma_n=\sum_{i=Z+1}^{A}\ket{\tilde\psi_i}\bra{\tilde\psi_i},
\end{equation}
and we denote $\tilde\gamma=(\tilde\gamma_p,\tilde\gamma_n)$; then, we show that
$$
\left[H_{\mu,\tilde\gamma},\tilde\gamma_\mu\right]=0
$$
for $\mu=p,n$. This implies
\begin{equation*}
\begin{array}{ll}
H_{p,\tilde\Psi}\tilde\psi_i=\varepsilon_{i}\tilde\psi_i & \mbox{for}\ 1\le i\le Z,\\[5pt]
H_{n,\tilde\Psi}\tilde\psi_i=\varepsilon_{i}\tilde\psi_i & \mbox{for}\ Z+1\le i\le A.
\end{array}
\end{equation*}

First of all, we observe that if $\tilde\Psi$ is a minimizer of (\ref{eqminproblem}), then the vector $\tilde\gamma=(\tilde\gamma_p,\tilde\gamma_n)$ is a minimizer of the energy 
\begin{eqnarray}\label{eqprojenergy}
\mathcal{E}(\gamma_p,\gamma_n)&=&\tr(H_0\gamma_p)+\tr(H_0\gamma_n)-\frac{g_\sigma^2}{8\pi}\int\int_{\mathbb{R}^3\times\mathbb{R}^3}\frac{\rho_s(x)\rho_s(y)}{|x-y|}e^{-m_\sigma|x-y|}\,dxdy\nonumber\\
&&+\frac{g_\omega^2}{8\pi}\int\int_{\mathbb{R}^3\times\mathbb{R}^3}\frac{\rho_0(x)\rho_0(y)}{|x-y|}e^{-m_\omega|x-y|}\,dxdy\nonumber\\
&&+\frac{g_\rho^2}{8\pi}\int\int_{\mathbb{R}^3\times\mathbb{R}^3}\frac{\rho_{00}(x)\rho_{00}(y)}{|x-y|}e^{-m_\rho|x-y|}\,dxdy\nonumber\\
&&+\frac{e^2}{8\pi}\int\int_{\mathbb{R}^3\times\mathbb{R}^3}\frac{\rho_{p}(x)\rho_{p}(y)}{|x-y|}\,dxdy
\end{eqnarray}
on $\Gamma^+_{Z,N}=\Gamma^+_{Z}\times\Gamma^+_{N}$ where
\begin{eqnarray*}
\Gamma^+_{Z}&=&\left\{\gamma_p\in \Gamma_Z; \gamma_p=\Lambda_{p,\gamma}^+\gamma_p\Lambda_{p,\gamma}^+ \right\},\\
\Gamma^+_{N}&=&\left\{\gamma_n\in \Gamma_N; \gamma_n=\Lambda_{n,\gamma}^+\gamma_n\Lambda_{n,\gamma}^+ \right\}.
\end{eqnarray*}
Next, we remind that
$$
H_{\mu,\gamma}=H^{+}_{\mu,\gamma}+H^{-}_{\mu,\gamma}
$$
with  $H^{+}_{\mu,\gamma}=\Lambda_{\mu,\gamma}^{+}H_{\mu,\gamma}\Lambda_{\mu,\gamma}^{+}$ and  $H^{-}_{\mu,\gamma}=\Lambda_{\mu,\gamma}^{-}H_{\mu,\gamma}\Lambda_{\mu,\gamma}^{-}$ for $\mu=p,n$. Then
$$
\left[H_{\mu,\gamma},\gamma_\mu\right]=\left[H^{+}_{\mu,\gamma},\gamma_\mu\right]+\left[H^{-}_{\mu,\gamma},\gamma_\mu\right].
$$
It is clear that $\Gamma^+_{Z,N}$ is a subset of
\begin{equation*}
\bar\Gamma_{Z,N}=\left\{\gamma=(\gamma_p,\gamma_n)\in\Gamma_Z\times\Gamma_N; \left[H^-_{p,\gamma},\gamma_p\right]=0, \left[H^-_{n,\gamma},\gamma_n\right]=0 \right\}
\end{equation*}
and, since $\tilde\gamma\in\Gamma^+_{Z,N}$, we obtain $\left[H^-_{\mu,\tilde\gamma},\tilde\gamma_\mu\right]=0$ for $\mu=p,n$. Thus, to conclude, we have to prove that $\left[H^+_{\mu,\tilde\gamma},\tilde\gamma_\mu\right]=0$ for $\mu=p,n$. We proceed by contradiction.\\ We suppose that  $\left[H^+_{p,\tilde\gamma},\tilde\gamma_p\right]$ and $\left[H^+_{n,\tilde\gamma},\tilde\gamma_n\right]$ are different from zero and we define 
\begin{eqnarray}\label{eqprojepsilonproton}
&&\tilde\gamma_p^{\varepsilon}=\mathcal{U}_p^\varepsilon\tilde\gamma_p\left(\mathcal{U}_p^\varepsilon\right)^{-1}:=\exp\left(-\varepsilon\left[H^+_{p,\tilde\gamma},\tilde\gamma_p\right]\right)\tilde\gamma_p\exp\left(\varepsilon\left[H^+_{p,\tilde\gamma},\tilde\gamma_p\right]\right),\\
\label{eqprojepsilonneutron}
&&\tilde\gamma_n^{\varepsilon}=\mathcal{U}_n^\varepsilon\tilde\gamma_n\left(\mathcal{U}_n^\varepsilon\right)^{-1}:=\exp\left(-\varepsilon\left[H^+_{n,\tilde\gamma},\tilde\gamma_n\right]\right)\tilde\gamma_n\exp\left(\varepsilon\left[H^+_{n,\tilde\gamma},\tilde\gamma_p\right]\right).
\end{eqnarray}
In particular,
\begin{eqnarray}\label{eqprojprotonepstilde}
\tilde\gamma_p^{\varepsilon}&=&\sum_{i=1}^{Z}\ket{\tilde\psi_i^{\varepsilon}}\bra{\tilde\psi_i^{\varepsilon}},\\
\label{eqprojneutronepstilde}
\tilde\gamma_n^{\varepsilon}&=&\sum_{i=Z+1}^{A}\ket{\tilde\psi_i^{\varepsilon}}\bra{\tilde\psi_i^{\varepsilon}}
\end{eqnarray}
with $\tilde\psi_i^{\varepsilon}=\mathcal{U}^{\varepsilon}_p\tilde\psi_i$ for $1\le i\le Z$ and $\tilde\psi_i^{\varepsilon}=\mathcal{U}^{\varepsilon}_n\tilde\psi_i$ for $Z+1\le i\le A$. \\
Using lemma \ref{lemimplicitfunction}, we construct $\gamma^{\varepsilon}=(\gamma_p^\varepsilon,\gamma_n^\varepsilon)$, small perturbation of $\tilde\gamma^{\varepsilon}=(\tilde\gamma_p^\varepsilon,\tilde\gamma_n^\varepsilon)$ such that 
\begin{eqnarray*}\label{eqprojconstraintproton}
&&\gamma^\varepsilon_p=\Lambda^+_{p,\gamma^\varepsilon}\gamma^\varepsilon_p\Lambda^+_{p,\gamma^\varepsilon},\\
\label{eqprojconstraintneutron}
&&\gamma^\varepsilon_n=\Lambda^+_{n,\gamma^\varepsilon}\gamma^\varepsilon_n\Lambda^+_{n,\gamma^\varepsilon}.
\end{eqnarray*}
We remark that $\gamma^\varepsilon\in \Gamma^+_{Z,N}$ and 
\begin{eqnarray}\label{eqprojprotoneps}
\gamma_p^{\varepsilon}&=&\sum_{i=1}^{Z}\ket{\phi_i^{\varepsilon}}\bra{\phi_i^{\varepsilon}},\\
\label{eqprojneutroneps}
\gamma_n^{\varepsilon}&=&\sum_{i=Z+1}^{A}\ket{\phi_i^{\varepsilon}}\bra{\phi_i^{\varepsilon}}
\end{eqnarray}
where 
\begin{eqnarray}\label{eqbaseprojprotoneps}
&&(\phi^\varepsilon_{1},\ldots,\phi^\varepsilon_{Z})=\Phi^\varepsilon_{p}=\Phi^{\varepsilon\,+}_{p}+\Phi^{\varepsilon\,-}_{p}+ O(\varepsilon^2)\\
\label{eqbaseprojneutroneps}
&&(\phi^\varepsilon_{Z+1},\ldots,\phi^\varepsilon_{A})=\Phi^\varepsilon_{n}=\Phi^{\varepsilon\,+}_{n}+\Phi^{\varepsilon\,-}_{n}+ O(\varepsilon^2)
\end{eqnarray}
with $\Phi^{\varepsilon\,+}_{\mu}=\Lambda_{\mu,\tilde\Psi^\varepsilon}^+\tilde\Psi_{\mu}^\varepsilon\bullet\left[\mathrm{Gram}_{L^2}(\Lambda_{\mu,\tilde\Psi^\varepsilon}^+\tilde\Psi_{\mu}^\varepsilon)\right]^{-1/2}$ for $\mu=p,n$. Finally, we remind that
\begin{equation}\label{eqbaseprojepstilde}
\tilde\Psi^\varepsilon_{\mu}=\Phi^{\varepsilon\,+}_{\mu}+\Lambda^-_{\mu,\tilde\Psi^{\varepsilon}}\tilde\Psi^{\varepsilon}_{\mu}\bullet B_{\mu}^{-1}+ O(\varepsilon^2)
\end{equation}
with $B_{\mu}=\left[\mathrm{Gram}_{L^2}(\Lambda_{\mu,\tilde\Psi^\varepsilon}^+\tilde\Psi_{\mu}^\varepsilon)\right]^{1/2}$ for $\mu=p,n$ (see the proof of lemma  \ref{lemimplicitfunction}).\\  
Then, to show that we have a contradiction, we want to prove that $$\mathcal{E}(\gamma^\varepsilon_p,\gamma^\varepsilon_n)<\mathcal{E}(\tilde\gamma_p,\tilde\gamma_n).$$ For this purpose, we calculate $\mathcal{E}(\gamma^\varepsilon_p,\gamma^\varepsilon_n)-\mathcal{E}(\tilde\gamma_p,\tilde\gamma_n)$; since $(\gamma^\varepsilon_p,\gamma^\varepsilon_n)$ is a small perturbation of $(\tilde\gamma_p,\tilde\gamma_n)$, we can write 
\begin{equation}\label{eqdiffenergy}
\mathcal{E}(\gamma^\varepsilon_p,\gamma^\varepsilon_n)-\mathcal{E}(\tilde\gamma_p,\tilde\gamma_n)=\tr\left(H_{p,\tilde\gamma}(\gamma^\varepsilon_p-\tilde\gamma_p)\right)+\tr\left(H_{n,\tilde\gamma}(\gamma^\varepsilon_n-\tilde\gamma_n)\right)+o(\varepsilon).
\end{equation}
To study the sign of (\ref{eqdiffenergy}), we remind that given an operator $T$ and an orthogonal projector $P$, we can consider the block decomposition of $T$ defined by
\begin{equation}\label{eqdecompositionproj1}
T=\left(\begin{array}{c|c}PTP & PT(1-P) \Bline\\\hline(1-P)TP\Tline & (1-P)T(1-P)\end{array}\right):=\left(\begin{array}{c|c}T_{++} & T_{+-} \Bline\\\hline T_{-+}\Tline & T_{--}\end{array}\right).
\end{equation}
Moreover, let $R$ be another orthogonal projector and consider $Q=R-P$; then $P+Q$ is a projector and 
\begin{eqnarray}\label{eqdecompositionproj2}
P+Q&=&(P+Q)^2 \nonumber\\
P+Q&=&P+PQ+QP+Q^2 \nonumber\\
Q^2&=&(1-P)Q-QP \nonumber\\
Q^2&=&(1-P)Q(1-P)-QP+(1-P)QP\nonumber\\
Q^2&=&Q_{--}-Q_{++}.
\end{eqnarray}
As a consequence, if $Q=O(\varepsilon)$, then $Q_{++}=O(\varepsilon^2)$ and $Q_{--}=O(\varepsilon^2)$.\\
Since 
$$
H_{\mu,\tilde\gamma}=\left(\begin{array}{c|c}H^+_{\mu,\tilde\gamma} & 0 \Bline\\\hline0 \Tline & H^-_{\mu,\tilde\gamma}\end{array}\right)
$$
for $\mu=p,n$, then we have
\begin{eqnarray}\label{eqdiffenergy1}
\mathcal{E}(\gamma^\varepsilon_p,\gamma^\varepsilon_n)&-&\mathcal{E}(\tilde\gamma_p,\tilde\gamma_n)=\tr\left(H_{p,\tilde\gamma}^+\Lambda_{p,\tilde\gamma}^+(\gamma^\varepsilon_p-\tilde\gamma_p)\Lambda_{p,\tilde\gamma}^+\right)\nonumber\\
&&+\tr\left(H_{p,\tilde\gamma}^-\Lambda_{p,\tilde\gamma}^-(\gamma^\varepsilon_p-\tilde\gamma_p)\Lambda_{p,\tilde\gamma}^-\right)
+\tr\left(H_{n,\tilde\gamma}^+\Lambda_{n,\tilde\gamma}^+(\gamma^\varepsilon_n-\tilde\gamma_n)\Lambda_{n,\tilde\gamma}^+\right)\nonumber\\
&&+\tr\left(H_{n,\tilde\gamma}^-\Lambda_{n,\tilde\gamma}^-(\gamma^\varepsilon_n-\tilde\gamma_n)\Lambda_{n,\tilde\gamma}^-\right)
+o(\varepsilon)\nonumber\\
&&:=T_p^++T_p^-+T_n^++T_n^-+o(\varepsilon).
\end{eqnarray}
First of all, we analyze the relation between $\Lambda^{\pm}_{\mu,\tilde\gamma}$, $\Lambda^{\pm}_{\mu,\tilde\gamma^{\varepsilon}}$ and  $\Lambda^{\pm}_{\mu,\gamma^{\varepsilon}}$ for $\mu=p,n$. Using (\ref{eqKato}), we obtain
\begin{eqnarray*}
&&\Lambda^{\pm}_{\mu,\tilde\gamma}=\Lambda^{\pm}_{\mu,\tilde\gamma^{\varepsilon}}+O(\varepsilon),\\
&&\Lambda^{\pm}_{\mu,\tilde\gamma}=\Lambda^{\pm}_{\mu,\gamma^{\varepsilon}}+O(\varepsilon),\\
&&\Lambda^{\pm}_{\mu,\tilde\gamma^{\varepsilon}}=\Lambda^{\pm}_{\mu,\gamma^{\varepsilon}}+O(\varepsilon).
\end{eqnarray*}
Then, if we take $P=\Lambda^{+}_{\mu,\gamma^{\varepsilon}}$, $Q=\Lambda^{+}_{\mu,\tilde\gamma}-\Lambda^{+}_{\mu,\gamma^{\varepsilon}}$ and we apply (\ref{eqdecompositionproj1})-(\ref{eqdecompositionproj2}), we obtain
$$
\Lambda^{-}_{\mu,\tilde\gamma}=\Lambda^{-}_{\mu,\gamma^{\varepsilon}}+\left(\begin{array}{c|c}O(\varepsilon^2) & O(\varepsilon)\Bline\\\hline O(\varepsilon)\Tline & O(\varepsilon^2)\end{array}\right)=\left(\begin{array}{c|c}O(\varepsilon^2) & O(\varepsilon)\Bline\\\hline O(\varepsilon)\Tline & 1+O(\varepsilon^2)\end{array}\right)
$$
for $\mu=p,n$.\\
Moreover, since $\gamma^{\varepsilon}\in \Gamma^+_{Z,N}$, we can write
$$
\gamma_\mu^\varepsilon=\left(\begin{array}{c|c}\gamma_{\mu++}^\varepsilon & 0\Bline\\\hline 0\Tline & 0\end{array}\right)
$$
and 
$$
\Lambda^{-}_{\mu,\tilde\gamma}\gamma_\mu^\varepsilon\Lambda^{-}_{\mu,\tilde\gamma}=\left(\begin{array}{c|c}O(\varepsilon^4) & O(\varepsilon^3)\Bline\\\hline O(\varepsilon^3)\Tline & O(\varepsilon^2)\end{array}\right).
$$
So we can conclude that $T^-_p=o(\varepsilon)$ and $T^-_n=o(\varepsilon)$.\\
Next, we remark that 
\begin{eqnarray*}
T^+_\mu&=&\tr\left(H_{\mu,\tilde\gamma}^+\Lambda_{\mu,\tilde\gamma}^+(\gamma^\varepsilon_\mu-\tilde\gamma_\mu^\varepsilon+\tilde\gamma_\mu^\varepsilon-\tilde\gamma_\mu)\Lambda_{\mu,\tilde\gamma}^+\right)\\
&=&\tr\left(H_{\mu,\tilde\gamma}^+\Lambda_{\mu,\tilde\gamma}^+(\gamma^\varepsilon_\mu-\tilde\gamma_\mu^\varepsilon)\Lambda_{\mu,\tilde\gamma}^+\right)+\tr\left(H_{\mu,\tilde\gamma}^+\Lambda_{\mu,\tilde\gamma}^+(\tilde\gamma_\mu^\varepsilon-\tilde\gamma_\mu)\Lambda_{\mu,\tilde\gamma}^+\right).
\end{eqnarray*}
To calculate $\tr\left(H_{\mu,\tilde\gamma}^+\Lambda_{\mu,\tilde\gamma}^+(\gamma^\varepsilon_\mu-\tilde\gamma_\mu^\varepsilon)\Lambda_{\mu,\tilde\gamma}^+\right)$, we consider the block decomposition of $\Lambda^{+}_{\mu,\tilde\gamma}-\Lambda^{+}_{\mu,\tilde\gamma^{\varepsilon}}$ for $P=\Lambda^{+}_{\mu,\tilde\gamma^{\varepsilon}}$. As before, we have
$$
\Lambda^{+}_{\mu,\tilde\gamma}=\Lambda^{+}_{\mu,\tilde\gamma^{\varepsilon}}+\left(\begin{array}{c|c}O(\varepsilon^2) & O(\varepsilon)\Bline\\\hline O(\varepsilon)\Tline & O(\varepsilon^2)\end{array}\right)=\left(\begin{array}{c|c}1+O(\varepsilon^2) & O(\varepsilon)\Bline\\\hline O(\varepsilon)\Tline & O(\varepsilon^2)\end{array}\right)
$$
for $\mu=p,n$.\\ 
Now, we observe that, in general, $\gamma^\varepsilon_\mu-\tilde\gamma_\mu^\varepsilon=O(\varepsilon)$ and, more precisely, $\Lambda_{\mu,\tilde\gamma^\varepsilon}^+(\gamma^\varepsilon_\mu-\tilde\gamma_\mu^\varepsilon)\Lambda_{\mu,\tilde\gamma^\varepsilon}^+=O(\varepsilon^2)$. Indeed, using the definitions from (\ref{eqprojprotonepstilde}) to (\ref{eqbaseprojepstilde}), we have
\begin{eqnarray*}
\Lambda_{\mu,\tilde\gamma^\varepsilon}^+(\gamma^\varepsilon_\mu-\tilde\gamma_\mu^\varepsilon)\Lambda_{\mu,\tilde\gamma^\varepsilon}^+&=&\sum_{i}\Lambda_{\mu,\tilde\gamma^\varepsilon}^+\left(\ket{\phi^\varepsilon_i}\bra{\phi^\varepsilon_i}-\ket{\tilde\psi^\varepsilon_i}\bra{\tilde\psi^\varepsilon_i}\right)\Lambda_{\mu,\tilde\gamma^\varepsilon}^+\\
&=&\sum_{i}\left(\ket{\phi^{\varepsilon\,+}_i}\bra{\phi^{\varepsilon\,+}_i}-\ket{\phi^{\varepsilon\,+}_i}\bra{\phi^{\varepsilon\,+}_i}\right) +O(\varepsilon^2)=O(\varepsilon^2).
\end{eqnarray*}
Then
$$
\Lambda^{+}_{\mu,\tilde\gamma}(\gamma^\varepsilon_\mu-\tilde\gamma_\mu^\varepsilon)\Lambda^{+}_{\mu,\tilde\gamma}=\left(\begin{array}{c|c}O(\varepsilon^2) & O(\varepsilon^3)\Bline\\\hline O(\varepsilon^3)\Tline & O(\varepsilon^4)\end{array}\right)
$$
and 
$$
T^+_\mu=\tr\left(H_{\mu,\tilde\gamma}^+\Lambda_{\mu,\tilde\gamma}^+(\tilde\gamma_\mu^\varepsilon-\tilde\gamma_\mu)\Lambda_{\mu,\tilde\gamma}^+\right)+o(\varepsilon).
$$
Next , we consider $\tilde\gamma_\mu^\varepsilon-\tilde\gamma_\mu$. By definition,
\begin{eqnarray*}
\tilde\gamma_\mu^\varepsilon-\tilde\gamma_\mu&=&\mathcal{U}^{\varepsilon}_\mu\tilde\gamma_\mu(\mathcal{U}^{\varepsilon}_\mu)^{-1}-\tilde\gamma_\mu\\
&=&\left(1-\varepsilon\left[H^+_{\mu,\tilde\gamma},\tilde\gamma_\mu\right]\right)\tilde\gamma_\mu\left(1+\varepsilon\left[H^+_{\mu,\tilde\gamma},\tilde\gamma_\mu\right]\right)-\tilde\gamma_\mu+o(\varepsilon)\\
&=&-\varepsilon\left[\left[H^+_{\mu,\tilde\gamma},\tilde\gamma_\mu\right],\tilde\gamma_\mu\right]+o(\varepsilon).
\end{eqnarray*}
Then
$$
T^+_\mu=-\varepsilon\tr\left(H_{\mu,\tilde\gamma}^+\left[\left[H^+_{\mu,\tilde\gamma},\tilde\gamma_\mu\right],\tilde\gamma_\mu\right]\right)+o(\varepsilon)
$$
for $\mu=p,n$ and
\begin{eqnarray}
\mathcal{E}(\gamma^\varepsilon_p,\gamma^\varepsilon_n)-\mathcal{E}(\tilde\gamma_p,\tilde\gamma_n)&=&-\varepsilon\sum_{\mu=p,n}\tr\left(H_{\mu,\tilde\gamma}^+\left[\left[H^+_{\mu,\tilde\gamma},\tilde\gamma_\mu\right],\tilde\gamma_\mu\right]\right)+o(\varepsilon)\nonumber\\
&=&2\varepsilon\sum_{\mu=p,n}\tr\left((H_{\mu,\tilde\gamma}^+\tilde\gamma_\mu)^2-(H_{\mu,\tilde\gamma}^+)^2\tilde\gamma_\mu^2\right)+o(\varepsilon)\nonumber\\
&=&2\varepsilon\sum_{\mu=p,n}\langle(H_{\mu,\tilde\gamma}^+\tilde\gamma_\mu)^*,H_{\mu,\tilde\gamma}^+\tilde\gamma_\mu\rangle-\langle H_{\mu,\tilde\gamma}^+\tilde\gamma_\mu,H_{\mu,\tilde\gamma}^+\tilde\gamma_\mu\rangle\nonumber\\
&&+o(\varepsilon)
\end{eqnarray}
where $\langle A ,B\rangle=\tr(A^*B)$ is the Hilbert–Schmidt inner product.\\
Then, using the Cauchy-Schwarz inequality, we obtain 
\begin{eqnarray*}
\left|\langle(H_{\mu,\tilde\gamma}^+\tilde\gamma_\mu)^*,H_{\mu,\tilde\gamma}^+\tilde\gamma_\mu\rangle\right|&\le& \langle(H_{\mu,\tilde\gamma}^+\tilde\gamma_\mu)^*,(H_{\mu,\tilde\gamma}^+\tilde\gamma_\mu)^*\rangle^{1/2}\langle H_{\mu,\tilde\gamma}^+\tilde\gamma_\mu,H_{\mu,\tilde\gamma}^+\tilde\gamma_\mu\rangle^{1/2}\\
&=&\langle H_{\mu,\tilde\gamma}^+\tilde\gamma_\mu,H_{\mu,\tilde\gamma}^+\tilde\gamma_\mu\rangle
\end{eqnarray*}
and
\begin{equation*}
\mathcal{E}(\gamma^\varepsilon_p,\gamma^\varepsilon_n)-\mathcal{E}(\tilde\gamma_p,\tilde\gamma_n)\le 0\,;
\end{equation*}
furthermore, the equality holds if and only if $(H_{\mu,\tilde\gamma}^+\tilde\gamma_\mu)^*=\pm H_{\mu,\tilde\gamma}^+\tilde\gamma_\mu$.\\
First, we consider the case $(H_{\mu,\tilde\gamma}^+\tilde\gamma_\mu)^*= H_{\mu,\tilde\gamma}^+\tilde\gamma_\mu$; this implies $\tilde\gamma_\mu H_{\mu,\tilde\gamma}^+=H_{\mu,\tilde\gamma}^+\tilde\gamma_\mu$ that means $\left[H_{\mu,\tilde\gamma}^+,\tilde\gamma_\mu\right]=0$. Then we have a contradiction.\\
Second, if $(H_{\mu,\tilde\gamma}^+\tilde\gamma_\mu)^*= -H_{\mu,\tilde\gamma}^+\tilde\gamma_\mu$, then
\begin{eqnarray*}
\tilde\gamma_\mu H_{\mu,\tilde\gamma}^++H_{\mu,\tilde\gamma}^+\tilde\gamma_\mu=0\\
\tilde\gamma_\mu H_{\mu,\tilde\gamma}^++\tilde\gamma_\mu H_{\mu,\tilde\gamma}^+\tilde\gamma_\mu=0\\
\tilde\gamma_\mu H_{\mu,\tilde\gamma}^+-H_{\mu,\tilde\gamma}^+\tilde\gamma_\mu=0
\end{eqnarray*}
that contradicts the hypothesis $\left[H_{\mu,\tilde\gamma}^+,\tilde\gamma_\mu\right]\neq0$ for $\mu=p,n$.\\ Finally, we can conclude that if $\left[H_{\mu,\tilde\gamma}^+,\tilde\gamma_\mu\right]\neq0$ for $\mu=p,n$, then we can construct $\gamma^\varepsilon\in\Gamma^{+}_{Z,N}$ such that 
\begin{equation*}
\mathcal{E}(\gamma^\varepsilon_p,\gamma^\varepsilon_n)-\mathcal{E}(\tilde\gamma_p,\tilde\gamma_n)< 0,
\end{equation*}
and thus we have a contradiction with the fact that $\tilde \gamma$ minimizes the energy on $\Gamma^+_{Z,N}$.\\
This implies that $\left[H_{\mu,\tilde\gamma}^+,\tilde\gamma_\mu\right]$ must be equal to zero and, as a consequence, $$\left[H_{\mu,\tilde\gamma},\tilde\gamma_\mu\right]=0$$
for $\mu=p,n$.\\
As a conclusion, if $g_\sigma,g_\omega,g_\rho$ and $e$ are sufficiently small, $\tilde\Psi$  is a solution of the equations (\ref{eqdiracyukawaproton}) and (\ref{eqdiracyukawaneutron}).

\appendix
\section{Proofs of lemma \ref{lemimplicitfunction} and corollary \ref{corimplicitfunction}}

In this section we give the proofs of lemma \ref{lemimplicitfunction} and corollary \ref{corimplicitfunction}.
\begin{proof}[Proof of lemma \ref{lemimplicitfunction}]
 Given an $M\times M$ matrix $B=(b_{ij})$, we denote $\Phi\bullet B$ the right action of $B$ on $\Phi=(\varphi_1,\ldots,\varphi_M)\in\left(L^2(\mathbb{R}^3)\right)^M$. More precisely,
\begin{equation*}\label{eqrightaction}
(\Phi\bullet B):=\left(\sum\limits_{i=1}^M b_{i1}\varphi_i,\ldots,\sum\limits_{i=1}^M b_{iM}\varphi_i\right)
\end{equation*}
and, by straightforward calculation, we obtain
$$
\mathrm{Gram}_{L^2}(\Phi\bullet B)=B^*\mathrm{Gram}_{L^2}(\Phi)B
$$
where $B^*$ denotes the conjugate transpose of $B$.\\
First of all, for $\mu=p,n$, we consider
\begin{equation}\label{eqpsirenormalized}
\tilde\Psi_\mu=\Psi_\mu\bullet G_\mu^{-1/2}
\end{equation}
and we observe that
\begin{eqnarray*}
\mathrm{Gram}_{L^2}\left(\tilde\Psi_{p}\right)=\mathbbm{1}_Z,\\
\mathrm{Gram}_{L^2}\left(\tilde\Psi_{n}\right)=\mathbbm{1}_N.
\end{eqnarray*}
Second, we define 
\begin{eqnarray}
\tilde\Phi_{p}^+&=&\Lambda^+_{p,\Psi}\tilde\Psi_{p} \bullet \left[\mathrm{Gram}_{L^2}\left(\Lambda^+_{p,\Psi}\tilde\Psi_{p}\right)\right]^{-1/2}\in \left(\Lambda^+_{p,\Psi}H^{1/2}\right)^Z, \\
\tilde\Phi_{n}^+&=&\Lambda^+_{n,\Psi}\tilde\Psi_{n} \bullet \left[\mathrm{Gram}_{L^2}\left(\Lambda^+_{n,\Psi}\tilde\Psi_{n}\right)\right]^{-1/2}\in \left(\Lambda^+_{n,\Psi}H^{1/2}\right)^N.
\end{eqnarray}
Remark that $\mathrm{Gram}_{L^2}\left(\Lambda^+_{p,\Psi}\tilde\Psi_{p}\right)$ and $\mathrm{Gram}_{L^2}\left(\Lambda^+_{n,\Psi}\tilde\Psi_{n}\right)$ are invertible matrices thanks to the hypothesis \ref{h2lemimplicitfunction} of the lemma.\\
Next, we look for $\left(\tilde\Phi_{p}^-,\tilde\Phi_{n}^-\right)\in \left(\Lambda^-_{p,\Psi}H^{1/2}\right)^Z\times  \left(\Lambda^-_{n,\Psi}H^{1/2}\right)^N$ such that, taking 
\begin{eqnarray*}
\Phi_{p}&=&l_{\tilde\Phi^+_p}(\tilde\Phi_{p}^-)\bullet G_p^{1/2}\\
\Phi_{n}&=&l_{\tilde\Phi^+_n}(\tilde\Phi_{n}^-)\bullet G_n^{1/2},
\end{eqnarray*}
we have 
\begin{eqnarray}\label{eqcondconstraintpro}
\Lambda^-_{p,\Psi}\Lambda^-_{p,\Phi}\Phi_{p}&=&0,\\  \label{eqcondconstraintneu}
\Lambda^-_{n,\Psi}\Lambda^-_{n,\Phi}\Phi_{n}&=&0,  
\end{eqnarray}
with $\Phi=(\Phi_{p},\Phi_{n})$ and $l_{\tilde\Phi^+_p}$, $l_{\tilde\Phi^+_n}$ defined by
\begin{equation*}
l_{\tilde\Phi^+_\mu}(\tilde\Phi_{\mu}^-):=\left(\tilde\Phi_{\mu}^++\tilde\Phi_{\mu}^-\right) \bullet \left[\mathrm{Gram}_{L^2}\left(\tilde\Phi_{\mu}^++\tilde\Phi_{\mu}^-\right)\right]^{-1/2}
\end{equation*}
for $\mu=p,n$.\\
We observe that $l_{\tilde\Phi^+_p}$ and $l_{\tilde\Phi^+_n}$ are smooth maps from $\left(\Lambda^-_{p,\Psi}H^{1/2}\right)^Z$ to $\left(H^{1/2}\right)^Z$ and  from $\left(\Lambda^-_{n,\Psi}H^{1/2}\right)^N$ to $\left(H^{1/2}\right)^N$ respectively; furthermore,
\begin{eqnarray*}
\mathrm{Gram}_{L^2}\left(\Phi_p\right)=G_p,\\ 
\mathrm{Gram}_{L^2}\left(\Phi_n\right)=G_n.
\end{eqnarray*}
Now, to prove the existence of $\tilde\Phi_{p}^-$ and $\tilde\Phi_{n}^-$, we apply the implicit function theorem. \\
We remark that the equations (\ref{eqcondconstraintpro}) and (\ref{eqcondconstraintneu}) can be written as $F(g,\tilde\Phi_{p}^-,\tilde\Phi_{n}^-)=0$ where $F$ is a nonlinear $\mathcal C^1$ operator and $g=(g_\sigma,g_\omega,g_\rho,e)$. In particular,
\begin{eqnarray*}
&&\Lambda^-_{\mu,\Psi}\Lambda^-_{\mu,\Phi}\Phi_{\mu}=\\
&&\Lambda^-_{\mu,\Psi}\Phi_{\mu}+\Lambda^-_{\mu,\Psi}\left(\frac{1}{2\pi}\int_{-\infty}^{+\infty}(H_{\mu,\Psi}-i\eta)^{-1}(H_{\mu,\Phi}-H_{\mu,\Psi})(H_{\mu,\Phi}-i\eta)^{-1}\Phi_{\mu}\,d\eta\right)\nonumber\\
\end{eqnarray*}
and
\begin{eqnarray*}
\Lambda^-_{\mu,\Psi}\Phi_{\mu}&=&\Lambda^-_{\mu,\Psi}\left(\tilde\Phi_{\mu}^++\tilde\Phi_{\mu}^-\right) \bullet \left[\mathbbm{1}+\mathrm{Gram}_{L^2}\left(\tilde\Phi_{\mu}^-\right)\right]^{-1/2}\bullet G_\mu^{1/2}\\
&=&\tilde\Phi_{\mu}^-\bullet \left[\mathbbm{1}+\mathrm{Gram}_{L^2}\left(\tilde\Phi_{\mu}^-\right)\right]^{-1/2}\bullet G_\mu^{1/2}.
\end{eqnarray*} 
Hence, we define
$$
F(g,\tilde\Phi_{p}^-,\tilde\Phi_{n}^-)=\left(\begin{array}{l}F_p(g,\tilde\Phi_{p}^-,\tilde\Phi_{n}^-)\\[5pt]F_n(g,\tilde\Phi_{p}^-,\tilde\Phi_{n}^-)\end{array}\right)
$$
where 
\begin{equation}\label{eqFp}
F_p(g,\tilde\Phi_{p}^-,\tilde\Phi_{n}^-)=\tilde\Phi_{p}^-\bullet \left[\mathbbm{1}+\mathrm{Gram}_{L^2}\left(\tilde\Phi_{p}^-\right)\right]^{-1/2}\bullet G_p^{1/2}+K_p(g,\tilde\Phi_{p}^-,\tilde\Phi_{n}^-),
\end{equation}
\begin{equation}\label{eqFn}
F_n(g,\tilde\Phi_{p}^-,\tilde\Phi_{n}^-)=\tilde\Phi_{n}^-\bullet \left[\mathbbm{1}+\mathrm{Gram}_{L^2}\left(\tilde\Phi_{n}^-\right)\right]^{-1/2}\bullet G_n^{1/2}+K_n(g,\tilde\Phi_{p}^-,\tilde\Phi_{n}^-)
\end{equation}
and
\begin{equation*}
K_\mu(g,\tilde\Phi_{p}^-,\tilde\Phi_{n}^-)=\Lambda^-_{\mu,\Psi}\frac{1}{2\pi}\int_{-\infty}^{+\infty}(H_{\mu,\Psi}-i\eta)^{-1}(H_{\mu,\Phi}-H_{\mu,\Psi})(H_{\mu,\Phi}-i\eta)^{-1}\Phi_{\mu}\,d\eta
\end{equation*}
for $\mu=p,n$.\\
Using the definitions (\ref{eqdiracyukawaproton}) and (\ref{eqdiracyukawaneutron}), we obtain $$K_p(0,\tilde\Phi_{p}^-,\tilde\Phi_{n}^-)=K_n(0,\tilde\Phi_{p}^-,\tilde\Phi_{n}^-)=0,$$ and then $F(0,0,0)=0$.\\
Now, to apply the implicit function theorem, we have to check that $$F:\mathbb R^4\times \left(\Lambda^-_{p,\Psi}H^{1/2}\right)^Z\times  \left(\Lambda^-_{n,\Psi}H^{1/2}\right)^N\rightarrow\left(\Lambda^-_{p,\Psi}H^{1/2}\right)^Z\times  \left(\Lambda^-_{n,\Psi}H^{1/2}\right)^N$$ is a $\mathcal C^1$ operator and $D_2F(0,0,0):=F_{\tilde\Phi_{p}^-,\tilde\Phi_{n}^-}(0,0,0)$ is an isomorphism.
We remark that
\begin{equation}\label{eqdiffF}
D_2F(0,0,0)(\chi,\tau)=\left(\begin{array}{l}\chi\bullet G_p^{1/2}\\[5pt]\tau\bullet G_n^{1/2}\end{array}\right),
\end{equation}
and then it is an isomorphism, since $G_p^{1/2}$ and $G_n^{1/2}$ are invertible matrices.
Proceeding as above, we can easily show that $F$ is well defined in $\left(\Lambda^-_{p,\Psi}H^{1/2}\right)^Z\times  \left(\Lambda^-_{n,\Psi}H^{1/2}\right)^N$.\\
Next, we have to prove that $F(g,\tilde\Phi_{p}^-,\tilde\Phi_{n}^-)$ is $\mathcal C^1$; by classical arguments, it is
enough to show that for $(\chi, \tau) \in \left(\Lambda^-_{p,\Psi}H^{1/2}\right)^Z\times  \left(\Lambda^-_{n,\Psi}H^{1/2}\right)^N$
\begin{eqnarray*}
&&\frac{\partial F_p(g,\tilde\Phi_{p}^-,\tilde\Phi_{n}^-)}{\partial \tilde\Phi_{p}^-}\chi\in \left(\Lambda^-_{p,\Psi}H^{1/2}\right)^Z,\\
&&\frac{\partial F_p(g,\tilde\Phi_{p}^-,\tilde\Phi_{n}^-)}{\partial \tilde\Phi_{n}^-}\tau\in \left(\Lambda^-_{p,\Psi}H^{1/2}\right)^Z,\\
&&\frac{\partial F_n(g,\tilde\Phi_{p}^-,\tilde\Phi_{n}^-)}{\partial \tilde\Phi_{p}^-}\chi\in \left(\Lambda^-_{n,\Psi}H^{1/2}\right)^N,\\
&&\frac{\partial F_n(g,\tilde\Phi_{p}^-,\tilde\Phi_{n}^-)}{\partial \tilde\Phi_{n}^-}\tau\in \left(\Lambda^-_{n,\Psi}H^{1/2}\right)^N,
\end{eqnarray*}
and we leave the details of this part to the reader. \\
Then, applying the implicit function theorem, we conclude that there exist $U\subset\mathbb{R}^4$, $V_p\subset \left(\Lambda^-_{p,\Psi}H^{1/2}\right)^Z$ and $V_n\subset  \left(\Lambda^-_{n,\Psi}H^{1/2}\right)^N$ neighborhoods of $0$, and a unique continuously differentiable function $f:U\rightarrow V_p\times V_n$ such that $F(g,f(g))=0$; that means that for $g_\sigma,g_\omega,g_\rho,e$ sufficiently small, there exists $(\tilde\Phi_{p}^{-},\tilde\Phi_{n}^{-})\in V_p\times V_n$ such that
\begin{eqnarray*}
\Lambda^-_{p,\Psi}\Lambda^-_{p,\Phi}\Phi_{p}&=&0,\\  
\Lambda^-_{n,\Psi}\Lambda^-_{n,\Phi}\Phi_{n}&=&0.  
\end{eqnarray*}
In particular, $U=\bar B(0,\gamma)$, $V_p= \bar B(0,\eta)$ and $V_n=\bar B(0,\eta)$ with $\gamma,\eta>0$ and from the proof of the implicit function theorem, we know that, fixed $\eta$, we can choose $\gamma$ such that $f:U\rightarrow V_p\times V_n$. Then we take $\eta$ and $\gamma$ such that $D_2F(g,\chi,\tau)$ is invertible $\forall (g,\chi,\tau)\in U\times V_p\times V_n$.\\
Now, we denote $B_p:=\left[\mathrm{Gram}_{L^2}\left(\Lambda^+_{p,\Psi}\tilde\Psi_{p}\right)\right]^{1/2}$ and we remark that
\begin{eqnarray*}
\tilde\Psi_p&=&\Lambda_{p,\Psi}^+\tilde\Psi_p+\Lambda_{p,\Psi}^-\tilde\Psi_p\nonumber\\
&=&\tilde\Phi^+_p\bullet B_p+\Lambda_{p,\Psi}^-\tilde\Psi_p.
\end{eqnarray*}
So we may write
$$
\tilde\Psi_p\bullet B_p^{-1}=\tilde\Phi^+_p+\Lambda_{p,\Psi}^-\tilde\Psi_p\bullet B_p^{-1}.
$$
As a consequence,
$$
l_{\tilde\Phi^+_p}(\Lambda_{p,\Psi}^-\tilde\Psi_p\bullet B_p^{-1})=(\tilde\Psi_p\bullet B_p^{-1})\bullet\left[\mathrm{Gram}_{L^2}\left(\tilde\Psi_p\bullet B_p^{-1}\right)\right]^{-1/2}.
$$
We can easily compute
$$
\mathrm{Gram}_{L^2}(\tilde\Psi_p\bullet B_p^{-1})=(B_p^*)^{-1} \mathrm{Gram}_{L^2}(\tilde\Psi_p)B_p^{-1}=(B_pB_p^* )^{-1}
$$
where $B_p^*$ denotes the conjugate transpose of $B_p$. Since $B_p$ is hermitian,
$$
\mathrm{Gram}_{L^2}(\tilde\Psi_p\bullet B_p^{-1})=(B_p^2 )^{-1}=(B_p^{-1})^2,
$$
$$
l_{\tilde\Phi^+_p}(\Lambda_{p,\Psi}^-\tilde\Psi_p\bullet B_p^{-1})=(\tilde\Psi_p\bullet B_p^{-1})\bullet(B_p^2 )^{1/2}=\tilde\Psi_p
$$
and
$$
l_{\tilde\Phi^+_p}(\Lambda_{p,\Psi}^-\tilde\Psi_p\bullet B_p^{-1})\bullet G_p^{1/2}=\Psi_p.
$$
Hence
\begin{equation}\label{eqdiffnormproton}
\left\|\Phi_p-\Psi_p\right\|_{(H^{1/2})^Z}=\left\|\left[l_{\tilde\Phi^+_p}(\tilde\Phi_p^-)-l_{\tilde\Phi^+_p}(\tilde\Psi_p^-\bullet B_p^{-1})\right]\bullet G_p^{1/2}\right\|_{(H^{1/2})^Z}
\end{equation}
with $\tilde\Psi^-_p=\Lambda_{p,\Psi}^-\tilde\Psi_p$. In the same way,
\begin{equation}\label{eqdiffnormneutron}
\left\|\Phi_n-\Psi_n\right\|_{(H^{1/2})^N}=\left\|\left[l_{\tilde\Phi^+_n}(\tilde\Phi_n^-)-l_{\tilde\Phi^+_n}(\tilde\Psi_n^-\bullet B_n^{-1})\right]\bullet G_n^{1/2}\right\|_{(H^{1/2})^N}
\end{equation}
with $\tilde\Psi^-_n=\Lambda_{n,\Psi}^-\tilde\Psi_n$ and $B_n:=\left[\mathrm{Gram}_{L^2}\left(\Lambda^+_{n,\Psi}\tilde\Psi_{n}\right)\right]^{1/2}$.\\
We remind that the maps $l_{\tilde\Phi^+_p}$ and $l_{\tilde\Phi^+_n}$ are smooth; then, to have an estimation of the norms (\ref{eqdiffnormproton}) and (\ref{eqdiffnormneutron}), it is enough to estimate
$$
\left\|\tilde\Phi_p^--\tilde \Psi_p^-\bullet B_p^{-1}\right\|_{(H^{1/2})^Z}\ \mbox{and}\ \left\|\tilde\Phi_n^--\tilde\Psi_n^-\bullet B_n^{-1}\right\|_{(H^{1/2})^N}.
$$
Indeed, $\forall\varepsilon>0$, $\exists\delta_p,\delta_n>0$ such that
$$
\left\|\tilde\Phi_p^--\tilde\Psi_p^-\bullet B_p^{-1}\right\|_{(H^{1/2})^Z}\leq\delta_p \Rightarrow \left\|l_{\tilde\Phi^+_p}(\tilde\Phi_p^-)-l_{\tilde\Phi^+_p}(\tilde\Psi_p^-\bullet B_p^{-1})\right\|_{(H^{1/2})^Z}\leq\varepsilon
$$
and
$$
\left\|\tilde\Phi_n^--\tilde\Psi_n^-\bullet B_n^{-1}\right\|_{(H^{1/2})^N}\leq\delta_n \Rightarrow \left\|l_{\tilde\Phi^+_n}(\tilde\Phi_n^-)-l_{\tilde\Phi^+_n}(\tilde\Psi_n^-\bullet B_n^{-1})\right\|_{(H^{1/2})^N}\leq\varepsilon.
$$
Now, for $\tilde{\delta}$ small enough, $(\tilde\Psi_p^-\bullet B_p^{-1},\tilde\Psi_n^-\bullet B_n^{-1}) \in V_p\times V_n$; then $F(g, \tilde\Psi_p^-\bullet B_p^{-1},\tilde\Psi_n^-\bullet B_n^{-1})$ is differentiable and $D_2F(g ,\tilde\Psi_p^-\bullet B_p^{-1},\tilde\Psi_n^-\bullet B_n^{-1}):=Q$ is invertible $\forall g \in U$. \\
Using this fact, we can write 
$$
F(g,\tilde\Phi_{p}^{-},\tilde\Phi_{n}^{-})=F(g, \tilde\Psi_p^-\bullet B_p^{-1},\tilde\Psi_n^-\bullet B_n^{-1})+Q(\tilde\Phi^{-}-\tilde\Psi^{-}\bullet B^{-1})+u(g,\tilde\Phi_{p}^{-},\tilde\Phi_{n}^{-})
$$
with
$$
\tilde\Phi^{-}=(\tilde\Phi_{p}^{-},\tilde\Phi_{n}^{-})=f(g),\ \tilde\Psi^{-}=(\tilde\Psi_{p}^{-},\tilde\Psi_{n}^{-}), \ B=\left(\begin{array}{cc}B_p& 0\\ 0& B_n\end{array}\right)
$$
and
\begin{equation}\label{eqlim}
\lim_{y\rightarrow\tilde\Psi^{-}\bullet B^{-1}} \frac{\left\|u(g,y)\right\|_{(H^{1/2})^A}}{\left\|y-\tilde\Psi^{-}\bullet B^{-1}\right\|_{(H^{1/2})^A}}=0,
\end{equation}
and this implies
$$
(\tilde\Phi^{-}-\tilde\Psi^{-}\bullet B^{-1})=-Q^{-1}F(g, \tilde\Psi_p^-\bullet B_p^{-1},\tilde\Psi_n^-\bullet B_n^{-1})-Q^{-1}u(g,\tilde\Phi_{p}^{-},\tilde\Phi_{n}^{-}).
$$
Moreover, thanks to (\ref{eqlim}), we know that there exists $\bar\delta>0$, such that 
$$
\left\|u(g,y)\right\|_{(H^{1/2})^A}\le \frac{1}{2\|Q^{-1}\|}\left\|y-\tilde\Psi^{-}\bullet B^{-1}\right\|_{(H^{1/2})^A}
$$
if $\left\|y-\tilde\Psi^{-}\bullet B^{-1}\right\|_{(H^{1/2})^A}\le \bar \delta$.\\
Then, choosing $\eta\le\frac{\bar\delta}{2}$, we have
\begin{eqnarray*}
\left\|\tilde\Phi^{-}-\tilde\Psi^{-}\bullet B^{-1}\right\|_{(H^{1/2})^A}&\le& \|Q^{-1}\|\left \|F(g,\tilde\Psi_p^-\bullet B_p^{-1},\tilde\Psi_n^-\bullet B_n^{-1})\right\|_{(H^{1/2})^A}\\
&&+\frac{1}{2}\left\|\tilde\Phi^{-}-\tilde\Psi^{-}\bullet B^{-1}\right\|_{(H^{1/2})^A}\\
\end{eqnarray*}
and
\begin{equation}\label{eqdiffnormminus}
\left\|\tilde\Phi^{-}-\tilde\Psi^{-}\bullet B^{-1}\right\|_{(H^{1/2})^A}\le C\left \|\left(\begin{array}{c}\Lambda^-_{p,\Psi}\Psi_{p}\\ \Lambda^-_{n,\Psi}\Psi_{n}\end{array}\right)\right\|_{(H^{1/2})^A}\le C\tilde{\delta}.
\end{equation}
Finally, choosing $\tilde\delta\le\frac{\min(\delta_p,\delta_n)}{C}$, we obtain
\begin{equation}\label{eqdiffnormprotonepsilon}
\left\|\Phi_p-\Psi_p\right\|_{(H^{1/2})^Z}\le\varepsilon
\end{equation}
and
\begin{equation}\label{eqdiffnormneutronepsilon}
\left\|\Phi_n-\Psi_n\right\|_{(H^{1/2})^N}\leq\varepsilon.
\end{equation}
To conclude the proof of the lemma, we have to show that
\begin{eqnarray*}
\Lambda^-_{p,\Psi}:\mathrm{Im}\Lambda^-_{p,\Phi}\rightarrow \mathrm{Im}\Lambda^-_{p,\Psi}\\
\Lambda^-_{n,\Psi}:\mathrm{Im}\Lambda^-_{n,\Phi}\rightarrow\mathrm{Im}\Lambda^-_{n,\Psi}
\end{eqnarray*}
are one-to-one operators. 
We remark that 
\begin{eqnarray*}
\left\|\Lambda^-_{p,\Phi}\Lambda^-_{p,\Psi}-I_{\mathrm{Im}\Lambda^-_{p,\Phi}}\right\|&=&\left\| \Lambda^-_{p,\Phi}\Lambda^-_{p,\Psi}-\Lambda^-_{p,\Phi}\right\|=\left\|\Lambda^-_{p,\Phi}\left(\Lambda^-_{p,\Psi}-\Lambda^-_{p,\Phi}\right)\right\|\\
&\le&\left\|\Lambda^-_{p,\Phi}\right\|\left\|\Lambda^-_{p,\Psi}-\Lambda^-_{p,\Phi}\right\|\le\left\|\Lambda^-_{p,\Psi}-\Lambda^-_{p,\Phi}\right\|<1.\\
\end{eqnarray*}
As a consequence, $\Lambda^-_{p,\Phi}\Lambda^-_{p,\Psi}$ is an invertible operator and $\Lambda^-_{p,\Psi}$ is one-to-one from $\mathrm{Im}\Lambda^-_{p,\Phi}$ into $\mathrm{Im}\Lambda^-_{p,\Psi}$. In the same way, we can prove that $\Lambda^-_{n,\Psi}$ is one-to-one from $\mathrm{Im}\Lambda^-_{n,\Phi}$ into $\mathrm{Im}\Lambda^-_{n,\Psi}$.\\
In conclusion,
\begin{eqnarray*}
\Lambda^-_{p,\Phi}\Phi_{p}&=&0,\\
\Lambda^-_{n,\Phi}\Phi_{n}&=&0.  
\end{eqnarray*}
 \end{proof}
 
\begin{proof}[Proof of corollary \ref{corimplicitfunction}] To prove this corollary, we apply lemma \ref{lemimplicitfunction} to $\Psi^k$ for any $k\in\mathbb{N}$ and, to obtain (\ref{eqconvprotonk}) and (\ref{eqconvneutronk}), we use the inequalities (\ref{eqdiffnormproton}), (\ref{eqdiffnormneutron}) and (\ref{eqdiffnormminus}).\\
\end{proof} 


\subsection*{Acknowledgment}
The author would like to thank Professor Eric Séré and Professor Bernhard Ruf for helpful discussions.\\
This work was partially supported by the 
Grant ANR-10-BLAN 0101 of the French Ministry of Research.


\bibliographystyle{postersimona}
\bibliography{theseintronew}


\end{document}